\documentclass[USenglish,a4paper,11pt]{article}

\usepackage{babel}
\usepackage{microtype}
\usepackage{amsmath,amsthm,amsfonts,amssymb}
\usepackage{mathtools}
\usepackage{thmtools,thm-restate}
\usepackage{nicefrac}
\usepackage{graphicx}
\usepackage{booktabs}
\usepackage{multirow}
\usepackage{titling}
\usepackage{authblk}
\usepackage{xcolor}
\usepackage[colorlinks=true,linkcolor=blue,citecolor=teal]{hyperref}
\usepackage{fullpage}
\usepackage[sort&compress,numbers]{natbib}
\usepackage[nameinlink]{cleveref}

\graphicspath{{./figures/}}

\newtheorem{theorem}{Theorem}

\newtheorem{lemma}[theorem]{Lemma}

\NewDocumentEnvironment{informaltheorem}{m}
{\par\medskip\noindent\textbf{Theorem~\ref{#1} (Informal).}\hspace{0.35em}\itshape\ignorespaces}
{\par\medskip}

\makeatletter
\AddToHook{cmd/appendix/before}{\def\cref@section@alias{appendix}}
\AddToHook{cmd/appendix/before}{\def\cref@subsection@alias{appendix}}
\makeatother

\newcommand{\OPT}{\mathrm{OPT}}
\newcommand{\Cmax}{C_{\textnormal{max}}}
\newcommand{\minkSP}{\textsc{Min}$k$\textsc{SP}}

\date{October 29, 2026}

\title{Diverse Schedules for Makespan Minimization\texorpdfstring{\\}{ }on Unrelated Machines\thanks{This collaboration originated at the Bremen--Hamburg Workshop in Bremen in 2025.
We thank the organizers for providing an environment that fostered this collaboration.}}

\newcommand{\email}[1]{\texttt{#1}}

\author[1]{Sebastian Bruchhold\thanks{This work was done while Sebastian Bruchhold was funded by the Deutsche Forschungsgemeinschaft (DFG\@, German Research Foundation) under Germany's Excellence Strategy -- The Berlin Mathematics Research Center MATH+ (EXC-2046/1, EXC-2046/2, project ID\@: 390685689).}}

\affil[1]{Goethe University Frankfurt, Institute of Computer Science, Frankfurt am Main, Germany
\authorcr
\email{Bruchhold@em.uni-frankfurt.de}}

\author[2]{Matthias Mnich}
\author[2]{Karolina Tammemaa\thanks{Supported by the Deutsche Forschungsgemeinschaft (DFG, German Research Foundation), project MN {59/4-1}}}

\affil[2]{Hamburg University of Technology, Institute for Algorithms and Complexity, Hamburg, Germany
\authorcr
\email{\string{matthias.mnich, karolina.tammemaa\string}@tuhh.de}}

\author[3]{Morten Weber\thanks{Funded by the Deutsche Forschungsgemeinschaft (DFG, German Research Foundation) under Germany's Excellence Strategy – EXC 3036 The Martian Mindset – project number~533607631.}}
\author[3]{Michalis Xefteris\thanks{Supported by a Humboldt Research Fellowship for Postdoctoral Researchers.}}
\author[3]{Bart Zondervan\thanks{Supported by the Deutsche Forschungsgemeinschaft (DFG, German Research Foundation), project number~547924951.}}

\affil[3]{University of Bremen, Faculty of Mathematics and Computer Science, Germany
\authorcr
\email{\string{moweber, xefteris, bart\string}@uni-bremen.de}}

\begin{document}
\maketitle
\begin{abstract}
    We study makespan minimization on unrelated machines with the additional goal of constructing diverse schedules:
    Instead of outputting a single schedule, we seek a collection of schedules with provable makespan guarantees and substantially different job assignments. This setting is motivated by practical applications in which committing to a single schedule in advance may be undesirable. Processing times may be uncertain, and additional constraints or domain knowledge may become available only later, so providing a set of $k$ schedules of similar quality gives the scheduler meaningful alternatives.
    We measure pairwise diversity by the weighted Hamming distance between job assignments and study the two standard objectives in the diverse-solutions literature: max-sum diversity, which maximizes the total pairwise distance between schedules, and max-min diversity, which maximizes the smallest pairwise distance.

    For the max-sum objective, we show that optimal diversity can be achieved in polynomial time without sacrificing the classic makespan approximation factor.
    For every feasible makespan target~$T$, we can compute~$k$ schedules of makespan at most~$2T$ whose max-sum diversity is optimal among all~$k$ schedules of makespan at most~$T$.
    Therefore, we match the classic factor-$2$ makespan guarantee of Lenstra, Shmoys, and Tardos (\emph{Math.\ Prog.}, 1990) and preserve optimal diversity.

    For the max-min objective, we show that the diversity problem is inherently harder.
    Namely, we show that, even without makespan constraints, deciding whether three schedules of given pairwise diversity exist is $\mathsf{NP}$-hard.
    We complete this with two complementary bicriteria approximation algorithms for every fixed~$k$ and $\varepsilon \in (0,1)$.
    The first preserves almost all of the optimal max-min diversity with makespan at most~$(3\lceil12/\varepsilon\rceil+1)T$.
    The second keeps the makespan arbitrarily close to the classic~$2$ guarantee while achieving a~$\frac{1}{2}(1-\varepsilon)$ fraction of the optimal max-min diversity.
\end{abstract}

\clearpage

\section{Introduction}
\label{sec:introduction}
Makespan minimization on unrelated machines is one of the most classic problems in scheduling.
Given a set~$M$ of~$m$ machines and a set~$J$ of~$n$ jobs, where each job~$j\in J$ has a machine-dependent processing time~$p_{ij}$ on machine~$i\in M$, one seeks an assignment, or schedule,~$\sigma:J\to M$ that minimizes the \emph{makespan}:
for a schedule~$\sigma$, the load of machine~$i$ is $L_i(\sigma)=\sum_{j\in J:\,\sigma(j)=i}p_{ij}$, and the makespan is $\Cmax(\sigma)=\max_{i\in M} L_i(\sigma)$.
The problem is denoted by~$R\|\Cmax$~\cite{GrahamLLR1979} and is well-known to be $\mathsf{NP}$-hard.
This hardness motivated the design of approximation algorithms, which compute in polynomial time a schedule~$\sigma$ whose makespan is within a factor~$\alpha$ of the optimum.
One of the most classic approximation algorithms in scheduling is the factor-$2$ approximation for~$R\|\Cmax$ due to~\cite{LenstraST1990}. 
Moreover, the same work shows that the problem is hard to approximate within any factor strictly smaller than~$\frac{3}{2}$, unless $\mathsf{P}=\mathsf{NP}$~\cite{LenstraST1990}.
Since its conception in 1990, no approximation factor strictly below~$2$ has been obtained for general unrelated machines.
Our main research question is whether one can retain such a guarantee when the output is required to contain not just one schedule, but several substantially different schedules.

Our motivation arises from settings in which committing to a single schedule in advance is undesirable. 
Processing times may be known only approximately, while additional constraints or application-specific information may become available only after candidate schedules have been computed. 
We therefore seek to preserve flexibility at the optimization stage by producing multiple schedules that are substantially different from one another. 
Such diversity enables a decision maker to incorporate information that was unavailable or unmodeled during optimization when selecting a final schedule. 
In this sense, diverse solution sets provide a natural mechanism for increasing robustness to uncertainty and unmodeled constraints.

This point of view is consistent with the broader diverse-solutions paradigm in combinatorial optimization.
A growing body of work studies how to compute multiple solutions that are not only good with respect to the original objective, but also far apart according to an appropriate distance measure~\cite{FellowsR2019, BasteFJMdOOPR2022, HanakaKKKKO2023}.
A similar perspective also arises in patient admission scheduling, where evolutionary diversity optimization has been used to compute diverse sets of high-quality schedules that provide alternatives under imperfect modeling~\cite{NikfarjamMNN2022}.

In this paper, we measure diversity at the level of job assignments: two schedules are far apart when many important jobs are assigned to different machines.
For example, a decision maker may later learn that certain job--machine assignments are desirable or undesirable and then choose the candidate schedule that best matches these preferences.
Optimizing this diversity creates a specific difficulty for unrelated-machine scheduling. 
Each schedule must satisfy its own load constraints, but the diversity objective couples all schedules through job assignments.
A change in one schedule may improve its distance from some schedules, decrease its distance from others, and at the same time violate a machine-load constraint.
Thus, diversity has to be built into the scheduling problem itself.

We study the following bicriteria problem.
Given a makespan target~$T$ and an integer~$k$, the goal is to compute~$k$ schedules with makespan close to~$T$ and large job-assignment diversity.
Our results separate the two standard diversity objectives: max-sum diversity maximizes the total pairwise distance in the collection, while max-min diversity maximizes the smallest pairwise distance.
For max-sum diversity, we retain the classic factor-$2$ makespan guarantee and achieve the optimal diversity among~$k$ schedules of makespan at most~$T$.
For max-min diversity, we show that the problem is computationally harder and design two complementary bicriteria algorithms.

\subsection{Our contributions}
We first formalize the diversity measure used throughout the paper. 
For two schedules~$\sigma_a$ and~$\sigma_b$, we define their pairwise diversity as the weighted Hamming distance 
\begin{align*}
  d(\sigma_a,\sigma_b) 
    \;=\; \sum_{j\in J} w_j \mathbf 1[\sigma_a(j)\ne \sigma_b(j)]  ,   
\end{align*}
where~$\mathbf 1[\cdot]$ is the indicator function and $w_j \ge 0$ is the \emph{diversity weight} (or \emph{weight} for short) of job~$j$, measuring the importance of scheduling~$j$ diversely.
Note that, since the diversity depends only on whether assignments differ, the machines are treated as labeled.
For~$k$ schedules, we follow the diverse-solutions literature and consider max-sum and max-min diversity~\cite{HanakaKKO2021, HanakaKKKKO2023}.
The max-sum and max-min objectives maximize, respectively, the total pairwise diversity and the minimum pairwise diversity, and are defined by
\begin{align*}
  d_\Sigma(\sigma_1,\dots,\sigma_k)
    \;=\; \sum_{1\le a<b\le k} d(\sigma_a,\sigma_b)
  \qquad\text{and}\qquad
  d_{\min}(\sigma_1,\dots,\sigma_k)
    \;=\; \min_{1\le a<b\le k} d(\sigma_a,\sigma_b) \enspace .
\end{align*}

The max-sum objective measures the total spread of the collection of schedules, and an optimal solution might include identical schedules or consist of a few well-separated clusters of nearly identical schedules. To counter such behavior, the max-min objective is stricter and requires every pair of schedules to be far apart.
An example for both objectives is depicted in~\Cref{fig:diversity}.

\begin{figure}
  \centering
  \includegraphics[width=0.8\linewidth]{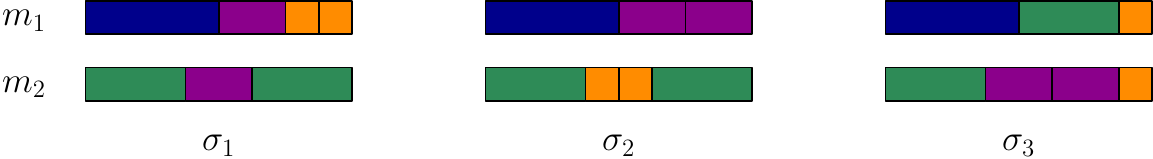}
  \caption{Three schedules~$\sigma_1$, $\sigma_2$, and~$\sigma_3$ of seven jobs on two machines~$m_1$ and~$m_2$.
  Jobs of the same color have the same processing time.
  All diversity weights are equal to~$1$.
  Then, the diversity of~$\sigma_1$ and~$\sigma_2$ is $d(\sigma_1, \sigma_2) = 3$,
  the diversity of~$\sigma_2$ and~$\sigma_3$ is $d(\sigma_2, \sigma_3) = 4$,
  and the diversity of~$\sigma_1$ and~$\sigma_3$ is $d(\sigma_1, \sigma_3) = 3$.
  Thus, $d_{\min}(\sigma_1, \sigma_2, \sigma_3) = 3$,
  while $d_\Sigma(\sigma_1, \sigma_2, \sigma_3) = 10$.}\label{fig:diversity}
\end{figure}

Diverse scheduling has two competing objectives, makespan and diversity, and our results approximate their Pareto tradeoff.
We state the guarantees for a makespan target~$T$, but they also admit the dual view of fixing a diversity target and minimizing the required makespan.
Given a target makespan~$T$, all our guarantees compare the diversity of the output schedules with the best diversity attainable by~$k$ schedules of makespan at most the original target~$T$, even though the output schedules may have larger makespan (see the discussion at the end of \Cref{sec:preliminaries}). 
Our results are summarized in \Cref{tab-contributions}.

\paragraph*{Max-sum diversity.}
Our first result is proved in \Cref{sec:max-sum-diversity}. 
We begin in \Cref{subsec:max-sum-warmup} by introducing a natural labeled $k$-tuple LP formulation, in which each of the $k$ schedules has a label and each variable assigns a job to a $k$-tuple of machines. 
This formulation can be viewed as a special case of the $k$-sided placement problem of \cite{KorupoluMRT2015}, and already gives optimal max-sum diversity with makespan~$(k+1)T$ in polynomial time for fixed~$k$. 
Our main result significantly improves the makespan guarantee and matches the best-known factor-$2$ approximation for $R\|\Cmax$~\cite{LenstraST1990}, while preserving exact max-sum diversity and running in time polynomial in the input size and~$k$. 
To do this, we rewrite the max-sum objective in terms of multiplicities. 
Occurrence-count arguments for weighted Hamming distance have already appeared in the diverse-solutions literature~\cite{HanakaKKKKO2023}.
In our work, however, multiplicities count job--machine assignments with unrelated processing times, so they must satisfy load constraints and be decomposed into actual schedules. We formulate a multiplicity LP in \Cref{subsec:max-sum-multiplicity-lp}. 
It tracks how many schedules assign each job to each machine, uses ranked copies to linearize the same-machine penalty, and keeps aggregate load constraints over the~$k$ schedules. 
We round this LP by iterative relaxation~\cite{LauRS2011} in \Cref{subsec:max-sum-rounding}, obtaining integral multiplicities with a small aggregate-load violation.
The remaining challenge is that this formulation controls only the aggregate load over all $k$ schedules, and an arbitrary realization of the rounded multiplicities could concentrate several large jobs in the same schedule. 
In \Cref{subsec:max-sum-decomposition}, we address this by decomposing the rounded multiplicities into labeled schedules, adapting the
slot construction by Shmoys and Tardos~\cite{ShmoysT1993} and combining it with bipartite edge coloring.
The resulting decomposition, together with an averaging argument that exploits the multiplicity structure, yields the per-schedule makespan guarantee $2T$.

\begin{informaltheorem}{thm:max-sum-two}
One can obtain $k$ schedules of makespan at most $2T$ with optimal max-sum diversity, in time polynomial in the input size and in~$k$.
\end{informaltheorem}

\paragraph*{Max-min diversity.}
We next turn to max-min diversity.
Unlike max-sum diversity, this objective must guarantee diversity for every pair of schedules. 
Thus, the multiplicity view is no longer enough: we need to know which pairs of schedules differ on each job. 
This obstruction is not only an artifact of our approach. 
\Cref{lem:max-min-hardness} in \Cref{app:max-min-hardness} shows that, given a diversity value~$Q$, it is $\mathsf{NP}$-hard to decide whether there exist three schedules with max-min diversity at least~$Q$, even when the makespan constraint is ignored.
For fixed $k$, we therefore give two complementary algorithms, one that preserves almost all of the optimum diversity, and one that keeps the makespan close to the factor-$2$ guarantee for $R\|\Cmax$.

\paragraph*{Near-optimal max-min diversity.}
In \Cref{subsec:near-optimal-max-min-diversity}, we give our first algorithm for max-min diversity. 
It keeps the labeled $k$-tuple LP and aims for near-optimal diversity. 
Assume that the diversity target is known through guessing.
For every pair of schedules, we guess the largest positive contributors to their pairwise distance (called \emph{heavy contributors}), together with their full $k$-tuple assignments. 
After guessing the diversity target and heavy contributors  for each pair of schedules, the algorithm solves a residual $k$-tuple LP and iteratively rounds it.
For each remaining job, we restrict the allowed tuples to use
at most $O(1/\varepsilon)$ distinct machines.
The main technical step is to prove that the residual LP remains
feasible despite this restriction.
The proof proceeds by contradiction, considering an LP point (a fractional assignment of jobs to $k$-tuple of machines) that
minimizes the sum of squared diversity shortfalls subject to
the job and load constraints. Here, a pair's \emph{shortfall} is the amount by which its diversity falls below its target, or zero if the target is met.
If some diversity target is unmet, we construct a point $y$
satisfying the job and load constraints.
The crucial idea is that, to construct $y$, the simultaneous
pairwise diversity requirements can be reduced to a single
weighted diversity target, using pairwise shortfalls as weights.
Meeting this weighted target ensures that moving slightly toward
$y$ decreases the sum of squared shortfalls, contradicting minimality.
We then round the residual LP by iterative relaxation.

\begin{informaltheorem}{thm:near-optimal-revised}
For fixed $k\ge 2$ and $\varepsilon\in (0,1)$, one can obtain $k$ schedules of make\-span at most $(3\lceil12/\varepsilon \rceil+1)T$ with $(1-\varepsilon)$-approximate max-min diversity, in polynomial time.
\end{informaltheorem}

In \Cref{app:near-optimal-max-min-small-k}, we also give an
alternative algorithm achieving the same diversity guarantee
with makespan at most $(2k+1+\varepsilon)T$.
Its run time is polynomial for fixed~$k$ and~$\varepsilon$.
This alternative bound can improve the makespan guarantee when $k$ is small.

\paragraph*{Low-makespan max-min diversity.}
In \Cref{subsec:2eps-apx-max-min-div}, we design our second algorithm for max-min diversity. Here we keep the makespan close to the factor-$2$ guarantee for $R\|\Cmax$.
The algorithm builds the schedules sequentially.
Using a metric argument, we show that at each step there exists a schedule from an optimal family that is at distance at least half the optimal max-min diversity from every schedule selected so far.
Thus, adding the next schedule becomes an unrelated-machine assignment problem with load upper bounds and one covering lower bound for each previously constructed schedule.
Related LPs appear in school-assignment-type problems~\cite{SanthiniMNS2025}.
We round this LP by first guessing heavy contributors and then applying iterative relaxation inspired by the rounding approach of Korupolu et al.~\cite{KorupoluMRT2015}.

\begin{informaltheorem}{thm:max-min-2-makespan}
For fixed $k\ge 2$ and $\varepsilon\in(0,1)$, one can obtain $k$ schedules of makespan at most $(2+\varepsilon)T$ with $\frac{1}{2}(1-\varepsilon)$-approximate max-min diversity, in polynomial time.
\end{informaltheorem}

\begin{table}[t]
 \centering
 \small
 \renewcommand{\arraystretch}{1.25}
 \begin{tabular}{@{}lccll@{}}
 \toprule
 \textbf{Objective} & \textbf{Diversity guarantee} & \textbf{Makespan guarantee} & \textbf{Run time} & {} \\
 \midrule
 \textbf{Max-sum}
 & $\OPT_k^\Sigma(T)$
 & $2T$
 & Polynomial time
 & (\Cref{thm:max-sum-two}) \\
 \midrule
 \multirow{3}{*}{\textbf{Max-min}}
 & $(1-\varepsilon)\OPT_k^{\min}(T)$
 & $(3 \lceil 12/\varepsilon \rceil+1)T$
 & \multirow{3}{*}{\begin{tabular}[c]{@{}l@{}}Polynomial time\\for fixed $k,\varepsilon$\end{tabular}}
 & (\Cref{thm:near-optimal-revised}) \\
 & $(1-\varepsilon)\OPT_k^{\min}(T)$
 & $(2k+1+\varepsilon)T$
 &
 &(\Cref{thm:max-min-ktuple}) \\
 & $\frac{1}{2}(1-\varepsilon)\OPT_k^{\min}(T)$
 & $(2+\varepsilon)T$
 &
 & (\Cref{thm:max-min-2-makespan}) \\
 \bottomrule
 \end{tabular}
 \caption{Summary of our main guarantees. $\OPT_k^\Sigma(T)$ denotes the optimum max-sum diversity, and $\OPT_k^{\min}(T)$ the optimum max-min diversity, among $k$ schedules of makespan at most $T$.}
 \label{tab-contributions}
 \vspace{-1em}
\end{table}

\subsection{Related work}
There is a long history of works that study combinatorial optimization problems with the goal of finding multiple ``diverse'' solutions:
For example, Krarup~\cite{Krarup1975} introduced the \textsc{Peripatetic Salesman} problem,
where multiple edge-disjoint Hamilton cycles in a graph need to be found that have minimum weight in sum.
Erkut et al.~\cite{ErkutN1989} survey problems where facility locations need to be computed that are maximally far apart.
Kuo et al.~\cite{KuoGD1993} consider the \textsc{Maximum Diversity} problem,
where a number of elements from a given set need to be selected such that the sum of their pairwise diversity coefficients is maximized, and show it to be~$\mathsf{NP}$-hard.
The computational intractability that is common for these diversity problems (see, e.g., by Hanaka et al.~\cite{HanakaKKKKO2023}) has motivated approaches via parameterized algorithms as well as approaches via approximation algorithms.

With a focus on the former approaches,
Fellows and Rosamond~\cite{FellowsR2019} recently introduced the \textsc{Diverse}~$X$ paradigm.
Here, $X$ can be any algorithmic problem and the proposal is to study ``diverse'' versions of~$X$, possibly with varying diversity measures, with the number of desired solutions as a parameter.
To give only a few examples, results of the study of parameterized algorithms for diverse problems encompass a framework for low-treewidth graphs~\cite{BasteFJMdOOPR2022}, a framework using \emph{odd sieving}~\cite{EibenKM2024}, and an algorithmic meta-theorem using cliquewidth~\cite{DrabikM2026}.

As Fellows and Rosamond~\cite{FellowsR2019} also point out, in applied settings, optimality is not so essential.
In line with this motivation, there is a growing body of literature on \emph{approximately} solving diverse problems,
including multiple frameworks~\cite{GaoGKTTY2022,HanakaKKKKO2023,GalvezGMPTV2025,AustrinBGLS2025},
and on selecting approximately diverse subsets of a \emph{given} set~\cite{CevallosEZ2017}.
While these have yielded approximation algorithms for diverse versions of, for example,
\textsc{Shortest Path}~\cite{GaoGKTTY2022},
\textsc{Interval Scheduling}~\cite{HanakaKKKKO2023},
\textsc{Knapsack}~\cite{GalvezGMPTV2025},
and \textsc{Interval Vertex Deletion}~\cite{AustrinBGLS2025},
the diverse version of makespan minimization on unrelated machines has not been studied.
To our knowledge, prior work on diverse scheduling problems thus far covers variants of \textsc{Diverse Interval Scheduling}~\cite{HanakaKKKKO2023}, including \textsc{Diverse Job Selection}~\cite{HanakaKKO2021}, and heuristics for hybrid flow-shop scheduling~\cite{GorenP2013}.
In \textsc{Interval Scheduling}, a maximum-cardinality subset of non-overlapping intervals needs to be selected from a given set of intervals,
and so the problem differs substantially from makespan minimization on unrelated machines.
G\"oren and Pierreval~\cite{GorenP2013} propose a heuristic framework for diverse scheduling and evaluate it on hybrid flow-shop makespan instances, without worst-case guarantees.
Our algorithms for unrelated machines provide approximation guarantees for both makespan and diversity.

\section{Preliminaries}
\label{sec:preliminaries}
We study nonpreemptive scheduling of jobs on unrelated machines.
Let~$M$ be a set of~$m$ machines and~$J$ a set of~$n$ jobs.
Each job~$j\in J$ has machine-dependent processing time $p_{ij} \in \mathbb{Q}_{\geq0}$ on machine~$i\in M$ and diversity weight~$w_j \in \mathbb{Q}_{\geq0}$.
A schedule $\sigma\colon J\rightarrow M$ assigns each job to one machine.
The load of machine~$i$ in schedule~$\sigma$ is~$L_i(\sigma) = \sum_{j\in J:\, \sigma(j)=i} p_{ij}$, and the makespan is~$\Cmax(\sigma) = \max_{i\in M} L_i(\sigma)$.
We write~$[k]=\{1,\dots,k\}$.

For machine-dependent load targets~$T_i$, we restrict each job~$j$ to its set of \emph{eligible} machines $\Gamma(j)=\{i\in M \mid p_{ij}\le T_i\}$.
No schedule respecting these load targets can assign~$j$ to a machine outside of~$\Gamma(j)$.
For a fixed makespan target $T$, this becomes $\Gamma(j)=\{i\in M \mid p_{ij}\le T\}$. This restriction is used in the algorithmic parts of the paper.
We also call a makespan target~$T$ \emph{feasible} if there exists at least one schedule~$\sigma$ with~$\Cmax(\sigma)\le T$. 

The pairwise diversity~$d$ defined in the introduction is a weighted Hamming distance with nonnegative weights.
If some weights are zero, it is formally a pseudometric.
In any case, it satisfies the triangle inequality.
Indeed, for each job~$j$, the indicator $\mathbf 1[\sigma_1(j)\ne\sigma_3(j)]$ is at most \mbox{$\mathbf 1[\sigma_1(j)\ne\sigma_2(j)]+\mathbf 1[\sigma_2(j)\ne\sigma_3(j)]$}, and summing with weights~$w_j\ge 0$ gives the claim.

For a makespan target~$T$, we define
\begin{align*}
  \OPT_k^\Sigma(T) \;=\; \max\Bigl\{ d_\Sigma(\pi_1,\dots,\pi_k) \;\Big|\; \Cmax(\pi_a) \leq T \ \forall a \in [k] \Bigr\} \enspace .
\end{align*}
Analogously, we define~$\OPT_k^{\min}(T)$ for the max-min diversity.

Let $d_\circ$, where $\circ\in\{\Sigma,\min\}$, denote the corresponding diversity objective and let $\OPT_k^\circ(T)$ denote the optimum value of this objective among $k$ schedules of makespan at most~$T$.
We call an algorithm an \emph{$(\alpha,\beta)$-bicriteria approximation} for the chosen diversity objective if, given a target makespan~$T$, it returns~$k$ schedules~$\sigma_1,\dots,\sigma_k$ such that $\Cmax(\sigma_a) \leq \alpha T$ for every $a\in [k]$ and $d_\circ(\sigma_1,\dots,\sigma_k) \geq \beta \OPT_k^\circ(T)$.

Throughout the paper we write $x^+ = \max\{x, 0\}$ for a real number $x$.

\paragraph*{Comparison benchmark.}
All bicriteria guarantees in this paper compare the diversity of the output schedules to the optimum diversity achievable by $k$ schedules of makespan at most the original target~$T$. 
Thus, if an algorithm computes schedules of makespan at most~$\alpha T$, the diversity benchmark for the chosen objective remains~$\OPT_k^\circ(T)$, rather than~$\OPT_k^\circ(\alpha T)$.

\section{Hardness}

Since makespan minimization is $\mathsf{NP}$-hard even when producing a single schedule, it is not surprising that hardness persists when producing \emph{multiple} diverse schedules.
However, we prove that producing multiple diverse schedules is hard even in settings where the makespan target is essentially irrelevant.

In \Cref{subsec:fixed-target-hardness}, we show that computing two schedules of makespan at most~$T$ with optimal diversity at target~$T$ is $\mathsf{NP}$-hard even when finding non-diverse schedules of makespan at most $T$ is trivial. The same hardness applies to both max-sum and max-min diversity, since for $k=2$ the two objectives coincide.

In \Cref{subsec:further-hardness}, we show that computing three schedules of optimal max-min diversity is $\mathsf{NP}$-hard even in the absence of a makespan target.
Further, we observe that the related problem of selecting~$k$ maximally diverse schedules from a \emph{given} set of schedules is $\mathsf{W}[1]$-hard with respect to~$k$, for both max-sum and max-min diversity.

\subsection{Hardness at a fixed makespan target}
\label{subsec:fixed-target-hardness}

We give a simple reduction showing that optimizing diversity at a fixed makespan target is $\mathsf{NP}$-hard. 
Our hardness result holds even when all diversity weights are equal to $1$ and every job is eligible on only two machines.
Each job has the same processing time on both eligible machines, and so we establish hardness even in the restricted-assignment machine model (the graph-balancing machine model \cite{EbenlendrKS2014} on stars, in fact).
Moreover, finding two distinct feasible schedules of makespan at most the target is trivial for the instances produced by the reduction.
Thus, the hardness comes from the diversity objective.

\begin{restatable}{lemma}{fixedtargethardness}
\label{lem:fixed-target-hardness}
  Given a makespan target~$T$ and a diversity value $Q$, it is $\mathsf{NP}$-hard to decide whether there exist two schedules $\sigma_1, \sigma_2$ such that $\Cmax(\sigma_1),\Cmax(\sigma_2)\le T$ and $d(\sigma_1,\sigma_2)\ge Q$.
  The same hardness applies to both max-sum and max-min diversity, since for $k=2$ the two objectives coincide.
\end{restatable}

The reduction is from \textsc{Partition} and establishes weak $\mathsf{NP}$-hardness; see \Cref{app:fixed-target-hardness} for the proof.

\subsection{Further hardness results}
\label{subsec:further-hardness}

We give an additional hardness result specifically for max-min diversity.
In particular, for this objective, producing three maximally diverse schedules is $\mathsf{NP}$-hard even without a makespan constraint.

\begin{restatable}{lemma}{maxminhardness}
\label{lem:max-min-hardness}
  Given a diversity value $Q$, it is $\mathsf{NP}$-hard to decide whether there exist three schedules $\sigma_1, \sigma_2, \sigma_3$ such that $d_{\min}(\sigma_1, \sigma_2, \sigma_3) \ge Q$, even if the makespan constraint is ignored.
\end{restatable}

The reduction is again from \textsc{Partition} and establishes weak $\mathsf{NP}$-hardness.
We refer to \Cref{app:max-min-hardness} for the details.
Note that this result separates max-min diversity from max-sum diversity in terms of hardness.
This is because max-sum diversity for three schedules without makespan constraints can be formulated as a linear objective subject to assignment constraints and thus solved in polynomial time.

Also, even when a set of schedules \emph{is already given}, deciding whether a sufficiently diverse fixed-size subset exists is hard for both max-sum and max-min diversity:
With an observation similar to that of Hanaka et al.~\cite[Observation~2]{HanakaKKKKO2023},
the problem of selecting a size-$k$ subset of maximum diversity from a given set of schedules can be shown to be $\mathsf{W}[1]$-hard with respect to~$k$ by a reduction from \textsc{Set Packing}, even with unit jobs on two machines.
Intuitively, the input set of schedules can omit easy solutions, making the problem hard.

\section{Max-sum diversity}
\label{sec:max-sum-diversity}
In this section, we present our bicriteria approximation algorithm for max-sum diversity. 
We begin in \Cref{subsec:max-sum-warmup} with the natural labeled $k$-tuple formulation, which keeps track of the machine chosen by each job in each of the $k$ schedules and yields a~$(k+1)$-makespan guarantee with optimal diversity. 
Our main result, in \Cref{subsec:max-sum-main-theorem}, shows that this labeled view is unnecessarily detailed: the max-sum diversity objective depends only on the \emph{multiplicities} with which each job is assigned to each machine across the~$k$ schedules, not on which schedule makes each assignment.

We exploit this observation through a multiplicity LP, introduced in \Cref{subsec:max-sum-multiplicity-lp}, which decides only how many of the schedules assign each job to each machine. 
The LP uses ranked copies to linearize the convex same-machine penalty and keeps only aggregate-load constraints summed over all schedules.
In \Cref{subsec:max-sum-rounding}, we round this LP by iterative rounding, obtaining integral multiplicities with a small aggregate-load violation.
Finally, in \Cref{subsec:max-sum-decomposition}, we decompose the multiplicities into~$k$ schedules using a slot construction,  bipartite edge coloring, and an averaging argument to bound the load of each output schedule.
Together, these steps yield~$k$ schedules of makespan at most~$2T$ whose max-sum diversity is at least~$\OPT_k^\Sigma(T)$.

\subsection{Warm-up: the labeled \texorpdfstring{$k$}{k}-tuple linear program}
\label{subsec:max-sum-warmup}

A natural first formulation for the problem keeps the labels of the $k$ schedules explicit. 
Concretely, we introduce for every job~$j\in J$ and every tuple~$\mathbf{i}=(i_1,\dots,i_k) \in \Gamma(j)^k$ a variable~$z^j_{\mathbf{i}}$ that equals~$1$ if and only if for every schedule~$a$ job~$j$ is assigned to machine~$i_a$.

For a job $j$ and a tuple $\mathbf{i}=(i_1,\ldots,i_k)$, let $\mu_i^j(\mathbf{i})=|\{a\in[k]\mid i_a=i\}|$ be the number of schedules in which the tuple assigns job $j$ to machine $i$.
Two schedules~$a<b$ contribute~$w_j$ to the diversity of job~$j$ exactly when~$i_a\neq i_b$, so the contribution of tuple~$\mathbf{i}$ to the total diversity is $w_j\left(\binom{k}{2} - \sum_{i\in \Gamma(j)}\binom{\mu_i^j(\mathbf{i})}{2}\right)$.
Here,~$\binom{\mu_i^j(\mathbf{i})}{2}$ counts the number of schedule pairs both using machine~$i$.

Since~$w_j \binom{k}{2}$ is a constant, maximizing total diversity is equivalent to minimizing the same-machine penalty $\phi^j_{\mathbf{i}} = w_j\sum_{i\in \Gamma(j)}\binom{\mu_i^j(\mathbf{i})}{2}$.
This gives a minimization integer program with objective $\sum_{j\in J} \sum_{\mathbf{i}\in \Gamma(j)^k} \phi_{\mathbf{i}}^j z_{\mathbf{i}}^j$ subject to assignment constraints $\sum_{\mathbf{i}\in \Gamma(j)^k} z^j_{\mathbf{i}}=1$ for every job~\mbox{$j\in J$}, and load constraints $\sum_{j\in J}\sum_{\mathbf{i}\in \Gamma(j)^k:\,i_a=i} p_{ij}z^j_{\mathbf{i}}\leq T$ for every machine~$i\in M$ and schedule \mbox{$a\in [k]$}.

This formulation is a special case of the minimization $k$-sided placement problem (\minkSP), the~$k$ sides are the~$k$ schedules, the nodes on each side are copies of the machines, and a tuple $(i_1,\dots,i_k)$ specifies the placement of one job on all~$k$ sides.
We recall \minkSP{} in~\Cref{subsec:min-k-sided-placement}.
Korupolu, Meyerson, Rajaraman, and Tagiku~\cite{KorupoluMRT2015} showed that~\minkSP\ admits a polynomial-time algorithm preserving the optimal linear cost while increasing every capacity by a factor of~$k+1$.\footnote{
    Korupolu, Meyerson, Rajaraman, and Tagiku note, based on a personal communication, that the Lov\'asz Local Lemma framework of Harris and Srinivasan for generalized scheduling problems~\cite{HarrisS2013,HarrisS2019} applies to \minkSP, yielding an $\mathcal{O}(\log k / \log\log k)$-capacity approximation.
    This would improve the makespan bound of this labeled $k$-tuple baseline from~$k+1$ to~$\mathcal{O}(\log k/\log\log k)$.
}
Instantiating this with cost~$\phi^j_{\mathbf{i}}$ and capacities~$T$ gives the following guarantee, whose proof is given in \Cref{app:max-sum}.

\begin{restatable}[Labeled $k$-tuple guarantee]{corollary}{ktupleresult}
\label{cor:k-tuple}
  Fix~$k\in \mathbb{N}_{\ge 2}$. There is a polynomial-time algorithm that, for every feasible makespan target~$T$, computes~$k$ schedules $\sigma_1,\dots,\sigma_k$, each with makespan at most~$(k+1)T$ and with $d_\Sigma(\sigma_1,\ldots,\sigma_k)\ge \OPT_k^\Sigma(T)$.
\end{restatable}

This baseline guarantee has two limitations.
First, the makespan guarantee~$(k+1)T$ grows with the number~$k$ of schedules.
Second, the formulation can have~$\Omega(nm^k)$ variables, so it is polynomial only for constant~$k$.
We now show that max-sum diversity has additional structure that lets us avoid the labeled formulation.

\subsection{Exact max-sum diversity with \texorpdfstring{$2$}{2}-approximate makespan}
\label{subsec:max-sum-main-theorem}

Let $p_{\max} = \max\{p_{ij} \mid j\in J,\ i\in \Gamma(j)\}$ be the largest processing requirement. 
Since every eligible pair satisfies $p_{ij}\le T$, we have~$p_{\max} \le T$. 
We now state our main result for max-sum diversity.

\begin{theorem}[Exact max-sum diversity with makespan~$2T$]
\label{thm:max-sum-two}
    There is an algorithm that, for every $k \in \mathbb{N}_{\geq 2}$ and every feasible target~$T$, in polynomial time in the input size and in $k$, computes schedules $\sigma_1, \dots, \sigma_k$ each with makespan at most~$T+p_{\max} \leq 2T$ and with diversity $d_\Sigma(\sigma_1,\ldots,\sigma_k)\ge \OPT_k^\Sigma(T)$.
\end{theorem}

This improves on the labeled~$k$-tuple formulation in two ways.
The makespan decreases from $(k+1)T$ to~$T+p_{\max}\leq 2T$, which is independent of~$k$.
Moreover, the algorithm has polynomial run time even when~$k$ is part of the input.
Both gains come from the same observation: the sum of diversities depends only on the multiplicities with which each job is assigned to each machine, not on which schedule makes each assignment.

The proof has three steps. 
First, we express the sum of diversities using only unlabeled multiplicities.
Second, we optimize and round these multiplicities using a multiplicity LP. 
Finally, we decompose the rounded multiplicities into~$k$ schedules by a slot-coloring argument.

\subsection{The multiplicity linear program}
\label{subsec:max-sum-multiplicity-lp}
Max-sum diversity does not depend on the labels of the schedules: it suffices to know, for every job~$j$ and machine~$i$, how many of the~$k$ schedules assign~$j$ to~$i$.
For schedules $\sigma_1,\dots,\sigma_k$, we let $\mu_{ij}=|\{a\in[k] \mid \sigma_a(j)=i\}|$ count the schedules assigning~$j$ to~$i$, so that~$\sum_{i\in \Gamma(j)} \mu_{ij} = k$.
Among the~$\binom{k}{2}$ pairs of schedules, exactly~$\sum_{i\in \Gamma(j)} \binom{\mu_{ij}}{2}$ place~$j$ on the same machine.
Therefore, job~$j$ contributes~$w_j \bigl(\binom{k}{2} - \sum_{i\in \Gamma(j)} \binom{\mu_{ij}}{2}\bigr)$ to the total diversity, and summing over all jobs gives
$d_\Sigma(\sigma_1,\ldots,\sigma_k) = \Bigl( \binom{k}{2}\sum_{j\in J} w_j \Bigr) - \sum_{j\in J} w_j\sum_{i\in \Gamma(j)}\binom{\mu_{ij}}{2}$.
We write $Q_0=\binom{k}{2} \sum_{j\in J} w_j$.
For the multiplicity vector~$\boldsymbol{\mu}=(\mu_{ij})_{j\in J, i\in \Gamma(j)}$, we define the same-machine penalty by
$\Phi(\boldsymbol \mu) = \sum_{j\in J} w_j\sum_{i\in \Gamma(j)}\binom{\mu_{ij}}{2}$.
We obtain~$d_\Sigma(\sigma_1,\ldots,\sigma_k)=Q_0-\Phi(\boldsymbol \mu)$. 
Since~$Q_0$ is independent of the schedules to compute, maximizing the sum of diversities is equivalent to minimizing the same-machine penalty~$\Phi(\boldsymbol \mu)$.

If a job $j$ is assigned to the same machine in multiple schedules, we think of these as \emph{copies} of $j$, and the $s^{\textnormal{th}}$ copy of $j$ is said to have \emph{rank} $s \in [k]$.
To linearize~$\Phi$, we use the identity~$\binom{\mu_{ij}}{2} = \sum_{s=1}^{\mu_{ij}} (s-1)$.
Thus, the $s^{\textnormal{th}}$ copy of job~$j$ on a fixed machine contributes penalty cost~$(s-1) w_j$.
For every eligible pair $(i,j)$ and every rank $s\in[k]$, we introduce a binary variable~$y_{ijs}$ representing selection of the~$s^{\textnormal{th}}$ copy of job~$j$ on machine~$i$.
We relax the binary constraints to obtain the multiplicity LP:
\begin{alignat}{3}
    \text{minimize} \quad
         \sum_{j\in J} w_j \sum_{i\in \Gamma(j)}\sum_{s\in [k]}
        &(s-1) y_{ijs} \tag{Mult-LP}
        \label{lp:multiplicity}
        \\[-0.5mm]
    \text{subject to}\qquad
        \sum_{i\in \Gamma(j)}\sum_{s\in [k]} y_{ijs} \;&=\; k
        \qquad &&\forall j\in J  , \notag\\[-0.5mm]
         \sum_{j\in J:\, i\in\Gamma(j)}\sum_{s\in [k]} p_{ij}y_{ijs} \;&\le\; kT
        \qquad &&\forall i\in M  , \notag\\[-0.5mm]
         0 \;\le\; y_{ijs} \;&\le\; 1
        \qquad &&\forall j\in J,\ i\in \Gamma(j),\ s\in[k]  \enspace . \notag
\end{alignat}
Note that, for fixed multiplicities, the cheapest ranks $1, \dots, \mu_{ij}$ preserve all loads and minimize the rank objective. Let~$Z_{\mathrm{mult}}$ denote the optimal value of~\eqref{lp:multiplicity}.
We prove the following lemma in \Cref{app:max-sum}.

\begin{restatable}{lemma}{multiplicityrelaxation}
\label{lem:multiplicity-lp-relaxation}
  The multiplicity LP~\eqref{lp:multiplicity} is a relaxation of the max-sum problem at makespan target~$T$. 
  In particular, $Z_{\mathrm{mult}}\le Q_0-\OPT_k^\Sigma(T)$.
\end{restatable}

\subsection{Rounding the multiplicity linear program}
\label{subsec:max-sum-rounding}
We iteratively round~\eqref{lp:multiplicity}.
At each step, we fix integral variables. 
If no such variable exists, the extreme-point structure yields a machine-load constraint with total remaining slack at most one, which we drop. 
This gives integral multiplicities with an aggregate load bound, which we later decompose into schedules in \Cref{subsec:max-sum-decomposition}.
The proof of the following lemma is deferred to \Cref{app:max-sum}.

\begin{restatable}[Multiplicity rounding]{lemma}{multiplicityrounding}
\label{lem:multiplicity-rounding}
  There is a polynomial-time algorithm that computes~$Y\in\{0,1\}^{m\times n\times k}$ such that
  $\sum_{i\in\Gamma(j)}\sum_{s\in [k]}Y_{ijs} = k$ for every $j\in J$, $\sum_{j\in J:i\in\Gamma(j)}\sum_{s\in [k]} p_{ij} Y_{ijs} \le kT+p_{\max}$ for every $i\in M$, and $\sum_{j\in J}w_j \sum_{i\in \Gamma(j)} \sum_{s\in [k]} (s-1) Y_{ijs} \le Z_{\mathrm{mult}}$.
\end{restatable}

We now define the rounded multiplicities~$\mu_{ij}=\sum_{s\in [k]} Y_{ijs}$.
By \Cref{lem:multiplicity-rounding},  $\sum_{i\in \Gamma(j)} \mu_{ij}=k$ for every job $j\in J$.
Each machine~$i\in M$ has load $L_i = \sum_{j\in J} p_{ij}\mu_{ij} = \sum_{j\in J} \sum_{s\in [k]} p_{ij} Y_{ijs} \le kT + p_{\max}$.

It remains to relate the penalty~$\Phi(\boldsymbol \mu)$ to the rounded objective.
For a fixed eligible pair~$(i,j)$, the~$\mu_{ij}$ ranks selected by~$Y$ need not be the lowest ranks~$1,2,\ldots,\mu_{ij}$.
Replacing them by the cheapest $\mu_{ij}$ ranks preserves the multiplicities, all load constraints, and can only decrease the rank objective. 
Therefore,
\begin{align}
        \Phi(\boldsymbol \mu)
        \;=\; \sum_{j\in J} w_j \sum_{i\in \Gamma(j)} \binom{\mu_{ij}}{2}
        \;\le\; \sum_{j\in J} w_j \sum_{i\in \Gamma(j)} \sum_{s\in [k]} (s-1) Y_{ijs}
        \;\le\; Z_{\mathrm{mult}} \enspace .
        \label{eq:phi-bound}
\end{align}

\subsection{Decomposing multiplicities into schedules}
\label{subsec:max-sum-decomposition}
We now turn the integer multiplicities~$\mu_{ij}$ into~$k$ schedules.
We build a bipartite multigraph between jobs and so-called ``machine slots'', properly edge-color it to obtain the schedules, and use an averaging argument over slots to bound the makespan of each schedule.

For every eligible pair~$(i,j)$ with~$\mu_{ij}>0$, we create~$\mu_{ij}$ copies of the pair~$(i,j)$, each of processing time~$p_{ij}$.
For each machine~$i$, we sort all copies involving~$i$ in nonincreasing order of processing time and cut this list into consecutive \emph{slots} of~$k$ copies each, with the last slot possibly smaller.
The ordering is inspired by the slot construction in the generalized-assignment rounding of Shmoys and Tardos~\cite{ShmoysT1993}.
In our setting, the slots are used to decompose integral multiplicities into~$k$ schedules via edge coloring, and the resulting load bound exploits this additional structure.

We construct a bipartite multigraph~$H$ with a vertex for each job on the left and a vertex for each machine slot on the right, adding one edge between job~$j$ and a slot of machine~$i$ for every copy of~$(i,j)$ in that slot.
Each job~$j$ has degree~$\sum_{i\in \Gamma(j)} \mu_{ij} = k$, and each machine slot has degree at most~$k$, as a slot holds at most~$k$ copies, so~$H$ has maximum degree~$k$.

By K\H{o}nig's edge-coloring theorem, $H$ admits a proper~$k$-edge-coloring~\cite{Konig1916, Schrijver2003} computable in polynomial time.
Each color defines one schedule.
If an edge of color~$a\in [k]$ is incident to job~$j$ and a slot of machine~$i$, then we set~$\sigma_a(j) = i$.
Since the coloring is proper and every job has degree exactly~$k$, each job has exactly one edge of each color, so every~$\sigma_a$ schedules all jobs. Moreover, for every pair~$(i,j)$, there are exactly~$\mu_{ij}$ edges from job~$j$ to slots of machine~$i$, and hence exactly~$\mu_{ij}$ colors assign~$j$ to~$i$. Thus,  job~$j$ is scheduled on machine~$i$ in exactly~$\mu_{ij}$ of the schedules.
Therefore,~$\sigma_1,\dots,\sigma_k$ realize the multiplicities~$\boldsymbol{\mu}$. \Cref{fig:slot-construction} illustrates the slot construction and the edge-coloring for one machine.

\begin{figure}[t]
  \centering
  \includegraphics[width=0.5\linewidth]{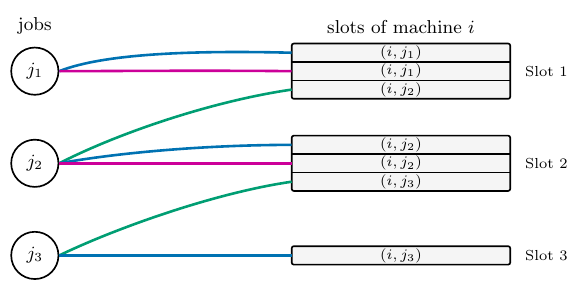}
  \caption{
  Slot construction for~$k=3$ and a fixed machine~$i$.
  The left vertices are jobs, the right vertices are the slots of machine~$i$, and each copy~$(i,j)$ becomes one edge from job~$j$ to the slot containing it. The edge colors correspond to three schedules~$\sigma_1,\sigma_2,\sigma_3$. Only the edges incident to slots of machine~$i$ are shown. In the full multigraph, the omitted incident edges of~$j_1$ and~$j_3$ go to slots of other machines.
  }
  \label{fig:slot-construction}
\end{figure}

We now bound the makespan of the constructed schedules.

\begin{lemma}[Slot decomposition]
\label{lem:slot-decomposition}
  For every machine $i \in M$ and color $a \in [k]$, the load of machine~$i$ in schedule $\sigma_a$ is at most $L_i(\sigma_a) \leq \frac{1}{k} \Bigl( L_i(\sigma_1)+\cdots+L_i(\sigma_k) \Bigr) +\left(1-\frac1k\right) p_{\max}$. 
\end{lemma}
\begin{proof}
  Fix a machine $i\in M$ and a color $a \in [k]$. 
  If~$i$ has no slots, then the load is zero and the claim holds, so let its slots be~$1,\dots,\ell$ in the sorted order.
  For slot~$u\in [\ell]$, let~$P_u$ be the total processing time of its copies and set~$F_u = P_u/k$. This is the average processing time of the copies in slot~$u$, except possibly for the last slot~$\ell$.
  For every color~$b\in[k]$, we define~$p_{bu}$ as the processing time of the copy selected by color~$b$ from slot~$u$, or~$0$ if it selects none.
  This is well-defined, because the coloring is proper, so each color selects at most one copy per slot.
  By the sorting, every copy in slot~$u$ has processing time at least as large as every copy in slot~$u+1$.

  Each full slot~$u<\ell$ holds exactly~$k$ copies.
  Since the coloring is proper and uses~$k$ colors, every color appears exactly once at such a slot.
  For the fixed color~$a$ we have \mbox{$P_u = p_{au} + \sum_{b\neq a} p_{bu}$} with~$p_{bu} \geq p_{a (u+1)}$ by the slot order. 
  Thus~$P_u \geq p_{au} + (k-1) p_{a (u+1)}$.
  It follows that
  $F_u = \frac1k \cdot P_u \ge \frac{1}{k} \cdot p_{au} + \left(1-\frac1k\right) p_{a (u+1)}$.
  Moreover, for the last slot we have~$F_\ell \geq \frac1k \cdot p_{a\ell}$.
  Summing these inequalities gives 
  $\sum_{u=1}^{\ell} F_u \ge \frac{1}{k} \cdot p_{a1} + \sum_{u=2}^{\ell}p_{au}$.
  This implies that~$\sum_{u\in [\ell]} p_{au} \leq \bigl(\sum_{u\in [\ell]} F_u\bigr) + (1-1/k) p_{a1}$.
  Recall that~$L_i = \sum_{j\in J} p_{ij}\mu_{ij} = \sum_{b \in[k]}L_i(\sigma_b)$, since the schedules realize the multiplicities. Thus, $\sum_{u\in [\ell]} F_u = L_i/k$ is the average load on machine~$i$ over all schedules.
  Moreover, we have~$p_{a1}\leq p_{\max}$.
  Thus, the load of machine~$i$ on schedule~$\sigma_a$ is 
  \begin{equation*}
    L_i(\sigma_a)
        \;=\; \sum_{u=1}^{\ell}p_{au}
        \;\le\; \left(\sum_{u=1}^{\ell} F_u\right) + \left(1-\frac{1} {k}\right)p_{a1} 
        \;\le\; \frac{L_i}{k}+\left(1-\frac{1}{k}\right) p_{\max}  \enspace . \qedhere
  \end{equation*}
\end{proof}

\subsection{Proof of \texorpdfstring{\Cref{thm:max-sum-two}}{Theorem \ref{thm:max-sum-two}}}
\label{subsec:max-sum-final-proof}
\begin{proof}
  By \Cref{lem:multiplicity-rounding}, the rounded multiplicities have aggregate load at most $kT+p_{\max}$ on every machine.
  Furthermore, by \Cref{lem:slot-decomposition}, every output schedule~$\sigma_a$ satisfies, for every machine~$i\in M$, the load bound
  \begin{align*}
    \begin{aligned}
        L_i(\sigma_a)
        &= \sum_{j\in J:\sigma_a(j)=i} p_{ij} 
        \le \frac{L_i}{k} + \left(1-\frac{1}{k}\right) p_{\max} 
        \le \frac{kT+p_{\max}}{k} + \frac{k-1}{k} p_{\max} 
        = T + p_{\max} \enspace .
    \end{aligned}
  \end{align*}
  In particular, this implies that~$\Cmax(\sigma_a) \leq T+p_{\max} \leq 2T$ for every schedule~$a\in [k]$.
    
  The decomposition preserves the multiplicities~$\mu_{ij}$. 
  Hence,~$d_\Sigma(\sigma_1,\ldots,\sigma_k) = Q_0-\Phi(\boldsymbol \mu)$.
  Using~\eqref{eq:phi-bound} and \Cref{lem:multiplicity-lp-relaxation}, we conclude that $d_\Sigma(\sigma_1,\ldots,\sigma_k) \ge Q_0-Z_{\mathrm{mult}} \ge \OPT_k^\Sigma(T)$. \qedhere
\end{proof}

\section{Max-min diversity}
\label{sec:max-min-diversity}
The diversity contribution of each job in the max-sum objective from the preceding section depends only on how often it is assigned to each machine across the~$k$ schedules, not on the labels of the schedules.
Hence, we optimized aggregate multiplicities and recovered the labeled schedules only at the end.
In contrast, for max-min diversity we must guarantee diversity for every pair of schedules.
Thus, we must know which pairs of schedules differ on each job.

The max-min objective is computationally harder than the max-sum objective because, given a diversity target~$Q$, it is already $\mathsf{NP}$-hard to decide whether there exist three schedules with $d_{\min}(\sigma_1,\sigma_2,\sigma_3)\ge Q$, even when dropping the makespan constraints (\Cref{lem:max-min-hardness}, see \Cref{app:max-min-hardness}).
Therefore, in this section, we give two complementary bicriteria approximation algorithms for max-min diversity.

First, in \Cref{subsec:near-optimal-max-min-diversity}, we achieve near-optimal diversity.
For fixed $k$ and $\varepsilon>0$, our labeled $k$-tuple algorithm achieves a $(1-\varepsilon)$-approximation to the optimum max-min diversity at the original target $T$, while increasing the makespan of the computed schedules to at most $(3\lceil12/\varepsilon\rceil+1)T$. 
After guessing the diversity target and heavy contributors to the diversity for each pair of schedules, the algorithm solves a residual $k$-tuple LP and iteratively rounds it.
For each remaining job, we restrict the allowed tuples to use at most $O(1/\varepsilon)$ distinct machines, and count the processing time of the job once on each machine used by its tuple.
The main technical step is to prove that the residual LP remains feasible despite this restriction. 

Second, in \Cref{subsec:2eps-apx-max-min-div}, we achieve the classic factor-$2$ makespan guarantee approximately.
The main idea is to build the schedules sequentially. 
Given the schedules already constructed, we view the construction of the next schedule as an unrelated-machine assignment problem with covering constraints, one per previously constructed schedule, each requiring the new schedule to be sufficiently diverse from it. 
For every fixed $k\in\mathbb N$ and $\varepsilon > 0$, we round this assignment problem and combine it with a metric argument to obtain schedules of makespan at most $(2+\varepsilon)T$ and a $\frac12(1-\varepsilon)$-approximation to the optimum max-min diversity.

\subsection{Near-optimal max-min diversity}
\label{subsec:near-optimal-max-min-diversity}

We first give an algorithm that obtains near-optimal max-min diversity. The details and full proofs of this subsection are given in \Cref{app:near-optimal-max-min}.
\begin{restatable}[Near-optimal max-min diversity]{theorem}{nearoptimalrevised}
    \label{thm:near-optimal-revised}
    Fix $k\in\mathbb N_{\ge2}$ and $\varepsilon\in(0,1)$.
    There is a deterministic algorithm that, for every feasible makespan
    target $T$, computes schedules $\sigma_1,\ldots,\sigma_k$ such that
    $\Cmax(\sigma_a)\le(3\lceil12/\varepsilon\rceil+1)T$, for every $a\in[k]$, and
    \text{$ d_{\min}(\sigma_1,\ldots,\sigma_k)\ge(1-\varepsilon)\OPT_k^{\min}(T)$.}
    The algorithm runs in time
    \((nm^k)^{O(k^4/\varepsilon)}\cdot|I|^{O(1)}\), where
    \(|I|\) is the input encoding length.
\end{restatable}

We start from the natural labeled $k$-tuple LP. 
For each job~$j$ and tuple $\mathbf{i} = (i_1,\ldots,i_k)\in\Gamma(j)^k$, we introduce a variable~$z^j_{\mathbf{i}}$ that indicates whether job~$j$ is placed on machine~$i_a$ in schedule~$a$ for all~$a\in [k]$.
Together, the tuples assigned to the jobs encode~$k$ labeled schedules.
For a pair $1\le a<b\le k$, let~$q^{a,b}_{\mathbf{i} ,j} = w_j \mathbf{1}[i_a\ne i_b]$; this is the contribution of job~$j$ to the diversity of pair~$(a,b)$ when the tuple~$\mathbf{i}$ places~$j$ on different machines in coordinates~$a$ and~$b$.
The diversity of pair $(a,b)$ is therefore the linear expression $\sum_{j\in J} \sum_{\mathbf{i}\in \Gamma(j)^k} q_{\mathbf{i},j}^{a,b} z_{\mathbf{i}}^j$. Thus, requiring minimum pairwise diversity at least $Q$ can be expressed by~$\binom{k}{2}$ linear covering constraints.

Set $\varepsilon_0=\varepsilon/3$, $r=\binom{k}{2}$, $h=\left\lceil\frac{2r}{\varepsilon_0}\right\rceil$, and $t=\left\lceil\frac4{\varepsilon_0}\right\rceil$.
We first guess the largest contributors to each pairwise diversity, so that the contributions remaining in the LP are small enough to control the rounding loss.
For a guessed diversity target $Q>0$, our algorithm guesses, for every pair of schedules, the $h$ jobs that are the largest positive diversity contributors in an optimal solution, or all positive contributors if there are fewer than $h$, together with their assigned tuples.
We discard guesses that are inconsistent or whose fixed jobs already violate a makespan constraint.

For a fixed remaining guess, let $G$ be the set of guessed jobs and let $\tau(j)=(\tau_1(j),\ldots,\tau_k(j))$ be the guessed tuple for each $j\in G$.
Define 
\begin{equation*}
Q_{a,b}^{\mathrm{fix}}
=\sum_{j\in G}w_j\mathbf1[\tau_a(j)\ne\tau_b(j)]
\qquad\text{and}\qquad
Q_{a,b}^{\mathrm{res}}=(Q-Q_{a,b}^{\mathrm{fix}})^+ \enspace .
\end{equation*}
A pair is \emph{active} when $Q_{a,b}^{\mathrm{res}}>0$, and \emph{inactive} otherwise.
The guessed tuples are fixed; they add at most $T$ to the load of each machine in each schedule.

The residual LP assigns tuples only to jobs in $J\setminus G$.
For each such job $j$, we keep only eligible tuples whose contribution to every active pair is at most~$Q/h$.
To obtain a makespan bound independent of $k$, we further restrict these tuples to use at most $t$ distinct machines.
Then, let $A_j$ denote the set of eligible tuples, so
\begin{equation*}
  A_j=
    \left\{
    \mathbf i\in\Gamma(j)^k \quad\middle|\quad
    |S(\mathbf i)|\le t,\quad
    q_{\mathbf i,j}^{a,b}\le Q/h
    \text{ for every active pair }(a,b)
    \right\} \enspace .
\end{equation*}
For tuple $\mathbf{i}=(i_1,\ldots,i_k)$ denote with $S(\mathbf{i})$ the set of different machines that appear in $\mathbf{i}$, so $S(\mathbf{i}) = \{i_1,\ldots,i_k\}$.
Note that $|S(\mathbf{i})|$ can be less than~$k$.
We define the following residual \(k\)-tuple LP:
\begin{equation}
\label{lp:res}\tag{Res-Tuple-LP}
\begin{aligned}
 \sum_{\mathbf{i}\in A_j}z_{\mathbf{i}}^j&=1
     &&\forall j\in J\setminus G,\\
 \sum_{j\in J\setminus G}\ \sum_{\mathbf{i}\in A_j:\,i\in S(\mathbf{i})}
                  p_{ij}z_{\mathbf{i}}^j&\le tT
     &&\forall i\in M,\\
 \sum_{j\in J\setminus G}\ \sum_{\mathbf{i}\in A_j}
                  q_{\mathbf{i},j}^{a,b}z_{\mathbf{i}}^j&\ge(1-\varepsilon_0)Q_{a,b}^{\mathrm{res}}
     &&\forall\text{ active }(a,b),\\
 z_{\mathbf{i}}^j&\ge0
     &&\forall j\in J\setminus G,\ \mathbf{i} \in A_j \enspace .
\end{aligned}
\end{equation}
The second constraint counts \(p_{ij}\) once whenever machine \(i\) appears in the tuple, regardless of how many coordinates use it. This charged load bounds the load on \(i\) in every schedule.
The guessed jobs are handled separately and add at most \(T\) to the load on each machine in each schedule.
The main technical step is showing that \eqref{lp:res} is feasible.
\begin{restatable}[Feasibility after guessing]{lemma}{feas}\label{lem:feas}
  Suppose that schedules $\pi_1,\ldots,\pi_k$ have makespan at most~$T$ and pairwise diversity at least~$Q$.
  For the correct heavy-contributor guess from these schedules, \eqref{lp:res} is feasible.
\end{restatable}
\begin{proof}[Proof sketch]
  We refer to $\pi_1, \dots, \pi_k$ as the \emph{reference family}.
  If $t\ge k$, then the tuples of the reference family directly establish feasibility.
  Otherwise, we reduce the simultaneous diversity requirements to a single weighted objective.
  Let $\mathcal P$ be the nonempty polytope of nonnegative points satisfying the job and load (but not necessarily the diversity) constraints of \eqref{lp:res}.
  Write $b_{a,b}=(1-\varepsilon_0)Q_{a,b}^{\mathrm{res}}$ for the diversity target of pair $(a,b)$ and $d_{a,b}(z)$ for its residual diversity at $z$.
  Choose $z\in\mathcal P$ minimizing
  \begin{equation*}
	\Phi(z) = \sum_{a<b}\bigl((b_{a,b}-d_{a,b}(z))^+\bigr)^2 \enspace .
  \end{equation*}
  If $\Phi(z)=0$, the lemma follows.
  Otherwise, give each pair weight $\lambda_{a,b}=(b_{a,b}-d_{a,b}(z))^+$.
  We will construct $y\in\mathcal P$ such that
  \begin{equation}
  \label{eq:weighted-witness-2}
	\sum_{a<b}\lambda_{a,b}d_{a,b}(y)
	\ge
	\sum_{a<b}\lambda_{a,b}b_{a,b} \enspace .
  \end{equation}
	Although $y$ does not need to meet every target individually, this inequality ensures that a sufficiently small move from $z$ toward $y$ decreases $\Phi$, contradicting the minimality of~$z$.
	
	The construction obtains this weighted diversity while controlling load. 
    Let the unguessed \emph{light} jobs be $L=\{j\in J\setminus G \mid w_j\le Q/h\}$. 
    For the correct guess, these jobs retain enough diversity in the reference schedules to meet every active residual
    target. 
    Define
    \begin{equation*}
    D = \sum_{a<b}\lambda_{a,b} \sum_{j\in L}w_j\mathbf{1}[\pi_a(j)\ne\pi_b(j)] \enspace .
    \end{equation*}
    Thus, $D$ is the reference family's pairwise diversity on the
    light jobs, aggregated with pair weights $\lambda_{a,b}$.
    We refer to this as \emph{weighted diversity}.
	We greedily select a common collection of~$t$ schedules from $\{\pi_1,\ldots,\pi_k\}$, starting with $\pi_1$, to determine the available machines for each light job.
	Consider the complete graph on the $k$ schedule labels with edge weights~$\lambda_{a,b}$.
	For each light job, a coloring using its available machines as colors determines its tuple: endpoints with different colors correspond to schedules that differ on that job.
	An optimal coloring (that is, with maximum total weight on dichromatic edges) therefore maximizes the job's weighted diversity.
	Each greedy step selects a reference schedule maximizing the total increase in this quantity over all light jobs.
	Every other unguessed job $j$ receives its constant-$\pi_1(j)$ tuple.
	
	Every resulting tuple belongs to $A_j$.
	Moreover, every used job--machine pair occurs in at least one of the $t$
	reference schedules, so the charged load of each machine is at most~$tT$.
	After~$t$ greedy selections, the resulting point $y$ has weighted diversity at least $(1-\varepsilon_0)D$.
	Since $D\ge\sum_{a<b}\lambda_{a,b}Q_{a,b}^{\mathrm{res}}$, this proves~\eqref{eq:weighted-witness-2} and yields the required contradiction.
\end{proof}

Having now established that the correct heavy-contributor guess makes \eqref{lp:res} feasible, we show next that we can round such a feasible LP to assign the unguessed jobs.

\begin{restatable}[Rounding after guessing]{lemma}{roundingrevised}\label{lem:rounding-revised}
    A deterministic polynomial-time algorithm exists that iteratively rounds feasible solutions of~\eqref{lp:res} such that exactly one tuple is assigned to each unguessed job, every machine has unguessed-job load at most~$3tT$ in every resulting schedule, and the unguessed jobs contribute diversity at least~$(1-\varepsilon_0)Q_{a,b}^{\mathrm{res}}-\varepsilon_0 Q$ to every active pair~$(a,b)$.
\end{restatable}

We use iterative rounding, ensuring that every iteration deletes at least one variable or drops a constraint.
In particular, in an extreme-point solution with strictly fractional variables, we drop the load constraint of machine~$i$ if $\sum_{j, \mathbf{i} \colon i \in S(\mathbf{i})} p_{i j} (1 - z_\mathbf{i}^j) \leq 2 t T$, or we drop the diversity constraint of an active pair~$(a, b)$ if $\sum_{j, \mathbf{i}} q_{\mathbf{i}, j}^{a, b} z_{\mathbf{i}}^j \leq \varepsilon_0 Q$.

\subsection{A \texorpdfstring{$(2+\varepsilon)$}{(2 + epsilon)}-makespan approximation for max-min diversity}
\label{subsec:2eps-apx-max-min-div}
In this section, we construct~$k$ schedules whose makespan is close to the target~$T$ while having bounded minimum pairwise diversity.
In particular, for constant~$k$ and~$\varepsilon$, our algorithm computes~$k$ schedules of makespan at most~$(2+\varepsilon)T$ and minimum pairwise diversity at least $\tfrac{1}{2}(1-\varepsilon) \OPT_k^{\min}(T)$.
The makespan factor matches, up to~$\varepsilon$, the best-known factor~$2$ for makespan minimization on unrelated machines~\cite{LenstraST1990}, which has resisted improvement for over three decades.

We construct the schedules sequentially.
Suppose that schedules~$\sigma_1,\dots,\sigma_r$ have already been constructed, and consider adding a new schedule~$\sigma$.
As before, the crucial observation is that the diversity of~$\sigma$ from each previous schedule is a linear function of the assignment of~$\sigma$.
Indeed, define coefficients~$q_{ij}^a = w_j \cdot \mathbf{1}[i \neq \sigma_a(j)]$ representing the contribution of job~$j$ to the diversity with respect to~$\sigma_a$, when schedule~$\sigma$ places~$j$ on machine~$i$ while schedule~$\sigma_a$ places it on a machine different from~$i$.
Summing over the jobs recovers exactly the pairwise diversity:~$d(\sigma,\sigma_a) = \sum_{j\in J} q_{\sigma(j),\,j}^a$.
Consequently, for a target diversity~$Q$, the new schedule~$\sigma$ has diversity at least~$Q$ from every previously constructed schedule if and only if for all~$a\in [r]$ we have~$\sum_{j\in J} q_{\sigma(j),\,j}^a \geq Q$.

Requiring~$\sigma$ to be far from all~$r$ previous schedules is thus captured by~$r$ covering constraints, one per previous schedule.
Together with the makespan bound, finding such a~$\sigma$ is an unrelated-machine scheduling problem with makespan and covering constraints.
We isolate this assignment problem and solve it in full generality, treating the covering contributions as arbitrary coefficients~$q_{ij}^a \in \mathbb{Q}_{\geq0}$ rather than specific diversity coefficients.
As before, let~$\Gamma(j) = \{i\in M\mid p_{ij}\leq T_i\}$ be the set of eligible machines for job~$j\in J$.
In particular, a machine~$i$ with~$p_{ij}>T_i$ cannot run~$j$ without violating its makespan target, so we restrict the assignment variables to~$x_{ij}$ with~$i\in\Gamma(j)$.

Assigning each job to a single eligible machine while meeting all makespan and covering targets is captured by the following LP, whose integral solutions are exactly the feasible assignments:
\begin{alignat}{3}
    \sum_{i\in \Gamma(j)} x_{ij} \;&=\; 1
    \qquad&& \forall j\in J, \tag{Cov-LP}\label{eq:max-min-div-LP}\\[-0.5mm]
    \sum_{j:\, i\in \Gamma(j)} p_{ij}x_{ij} \;&\le\; T_i \qquad&& \forall i\in M,\notag\\[-1mm]
    \sum_{j\in J}\sum_{i\in \Gamma(j)} q^a_{ij}x_{ij} \;&\ge\; Q_a \qquad&& \forall a\in [r],\notag\\[-0.5mm] x_{ij} \;&\ge\; 0
    \qquad&& \forall j\in J,\ i\in\Gamma(j) \enspace .\notag
\end{alignat}
The makespan constraints upper bound the load of each machine, while the~$r$ covering constraints lower bound the total contribution collected.
The main tool for establishing our bicriteria approximation for max-min diversity is the following proposition which we prove in~\Cref{app:covering-constraints} by combining a guessing argument with iterative rounding.
\begin{restatable}{proposition}{roundingcoveringconstraints}
\label{prop:rounding-covering-constraints}
  Fix $r \in \mathbb{N}_0$ and $\varepsilon\in (0,1)$.
  If~$\eqref{eq:max-min-div-LP}$ admits an integral feasible solution, then an integral assignment~$X$ can be computed in time~$(nm)^{\mathcal{O}(r^2/\varepsilon^2)}\cdot |I|^{O(1)}$, where $|I|$ is the input encoding length, satisfying~$\sum_{i\in \Gamma(j)} X_{ij} = 1$ for all~$j\in J$, $\sum_{j\in J:i\in \Gamma(j)} p_{ij} X_{ij} \leq (2+\varepsilon) T_i$ for all~$i\in M$, and~$\sum_{j\in J} \sum_{i\in \Gamma(j)} q_{ij}^a X_{ij} \geq (1-\varepsilon) Q_a$ for all~$a\in [r]$.
\end{restatable}

We now return to the $(2+\varepsilon,\tfrac12 (1-\varepsilon))$-bicriteria approximation for max-min diversity.
We construct~$k$ schedules sequentially by applying~\Cref{prop:rounding-covering-constraints} exactly $k$ times. For a diversity threshold~$Q\le \OPT_k^{\min}(T)$, using a metric argument we show that at each step there exists a schedule of makespan at most~$T$ and diversity at least~$\tfrac12 Q$ from each previously constructed schedule.
\begin{restatable}[Low-makespan max-min diversity]{theorem}{maxmintwomakespan}
\label{thm:max-min-2-makespan}
  Fix $k \in \mathbb{N}_{\ge 2}$ and $\varepsilon\in (0,1)$. There is an algorithm that, for every feasible makespan target~$T$, computes schedules~$\sigma_1,\dots,\sigma_k$ such that $\Cmax(\sigma_a) \leq (2+\varepsilon) T$ for every~$a\in [k]$ and~$d_{\min}(\sigma_1,\dots,\sigma_k) \geq \frac12 (1 - \varepsilon) \OPT_k^{\min}(T)$. The algorithm runs in time~$(nm)^{\mathcal{O}(k^2/\varepsilon^2)}\cdot |I|^{O(1)}$, where $|I|$ is the input encoding length. 
\end{restatable}
\begin{proof}[Proof sketch]
We first prove a threshold statement. Set $\varepsilon_0=\varepsilon/3$, fix a diversity threshold~$Q\le \OPT_k^{\min}(T)$, and suppose for now that such a value~$Q$ is given. 
We compute~$k$ schedules of makespan at most~$(2+\varepsilon_0)T$ and max-min diversity at least~$\frac12(1-\varepsilon_0)Q$.
We construct the schedules sequentially by applying~\Cref{prop:rounding-covering-constraints} exactly~$k$ times. In the first application, there are no covering constraints. Suppose that schedules~$\sigma_1,\dots,\sigma_{r-1}$ have already been constructed. Assume that each of them has makespan at most~$(2+\varepsilon_0)T$ and that every pair among them has diversity at least~$\frac12(1-\varepsilon_0)Q$.

Let~$\pi_1,\dots,\pi_k$ be an optimal solution, so $\Cmax(\pi_a)\le T$ for every~$a\in[k]$ and $d(\pi_a,\pi_{a'})\ge \OPT_k^{\min}(T)\ge Q$ for every distinct~$a,a'\in[k]$.
We claim that one of these schedules is at distance at least~$Q/2$ from every previously constructed schedule.
Suppose not. Then for every~$a\in[k]$, there exists some~$b(a)\in[r-1]$ such that $d(\pi_a,\sigma_{b(a)}) < \frac12 Q$. Since $r-1<k$, the pigeonhole principle implies that two distinct optimum schedules~$\pi_a$ and~$\pi_{a'}$ are mapped to the same constructed schedule~$\sigma_b$. By the triangle inequality, we get $d(\pi_a,\pi_{a'}) \le d(\pi_a,\sigma_b)+d(\sigma_b,\pi_{a'}) < \frac12 Q+\frac12 Q = Q$.
This contradicts $d(\pi_a,\pi_{a'})\ge \OPT_k^{\min}(T)\ge Q$. Therefore, there exists~$a\in[k]$ such that $d(\pi_a,\sigma_b)\ge \frac12 Q$ for every $b\in[r-1]$.

Fix such a schedule~$\pi_a$.
The schedule~$\pi_a$ has makespan at most~$T$ and diversity at least~$\tfrac12 Q$ from each of~$\sigma_1,\dots,\sigma_{r-1}$.
It is therefore an integral feasible solution of~\eqref{eq:max-min-div-LP} with makespan targets~$T_i=T$ for all machines~$i$ and covering targets~$\tfrac12 Q$, one covering constraint for each previously constructed schedule.
Applying~\Cref{prop:rounding-covering-constraints} with parameter~$\varepsilon_0$ gives a schedule~$\sigma_r$ with~$\Cmax(\sigma_r)\leq(2+\varepsilon_0)T$ and, for all~$b\in[r-1]$, we get $d(\sigma_b,\sigma_r)\geq (1-\varepsilon_0)\cdot \frac12 Q = \frac12(1-\varepsilon_0)Q$.
Thus~$\sigma_1,\dots,\sigma_r$ satisfy the same invariant.
After~$k$ steps, this proves the threshold statement for every~$Q\le \OPT_k^{\min}(T)$. In \Cref{app:low-makespan-max-min}, we complete the proof by guessing~$Q$ geometrically and absorbing the loss through the choice $\varepsilon_0=\varepsilon/3$. This gives the stated $\frac12(1-\varepsilon)\OPT_k^{\min}(T)$ guarantee and run time bound.
\end{proof}

\section{Future work}

In this work, we studied bicriteria approximation algorithms for scheduling on unrelated machines under the max-sum and max-min diversity objectives. The main open question that remains is whether max-min diversity admits a polynomial-time algorithm with constant approximation guarantees for both makespan and diversity when $k$ is part of the input. 
Further, future work can go beyond the standard notions of max-sum and max-min diversity.
For example, schedules could instead be compared using \emph{coverage diversity}, \emph{absolute-difference diversity}, or \emph{Venn diversity}~\cite{deBergLS2025,DrabikM2026}.

Although unrelated machines generalize models such as restricted assignment, related machines, and identical machines, from a practical perspective, other notions of diversity may be of interest especially for identical machines.
In this special case, our notion of diversity can consider two schedules maximally diverse even when only the labels of machines differ but the two schedules partition the jobs in the same way.
Thus, it is of interest to study notions of diversity that are invariant under machine relabeling.

\paragraph*{AI acknowledgment.}
Generative AI models (GPT-6 Astra and GPT-6.1 Sol) were used to search the literature, check proofs, and draft figures. GPT-6 Astra developed the algorithm and proof of \Cref{thm:near-optimal-revised} for near-optimal max-min diversity with a makespan approximation factor independent of~$k$, using as input the authors' existing result in~\Cref{thm:max-min-ktuple}, which guarantees makespan at most~$(2k+1+\varepsilon)T$. The authors take full responsibility for all content of the paper.

\bibliographystyle{abbrvnat}
\bibliography{literature-arxiv}

\clearpage

\appendix

\section{Hardness at a fixed makespan target}
\label{app:fixed-target-hardness}

\fixedtargethardness*
\begin{proof}
    We reduce from \textsc{Partition}. 
    Let $a_1, \dots, a_n$ be positive integers with total sum~$2B$, and also assume without loss of generality that $a_j \le B$ for every $j$.
    We create one job~$j$ for each~$a_j$, so the set of all jobs is $J = [n]$. 
    There is one machine~$m^\star$ and, for every job~$j$, one private dummy machine~$d_j$. 
    We set $p_j=a_j$ and define $p_{m^\star j}=p_{d_j j}=p_j$ and $p_{ij}=B+1$ for every other machine~$i$. Thus, at makespan target~$B$, job~$j$ is eligible exactly on~$m^\star$ and~$d_j$.
    All diversity weights are equal to~$1$. 
    We set the makespan target to~$T=B$ and the target diversity to~$Q=n$.
    
    If the \textsc{Partition} instance has a set~$S \subseteq J$ with~$\sum_{j\in S}a_j=B$, then we define two schedules as follows. 
    In the first schedule, assign the jobs of~$S$ to~$m^\star$ and all other jobs to their private dummy machines. 
    In the second schedule, we assign the jobs of~$J\setminus S$ to~$m^\star$ and the jobs of~$S$ to their private dummy machines. 
    Both schedules have makespan~$B$ and every job is assigned to different machines in the two schedules. 
    Hence $d(\sigma_1,\sigma_2)=n$.
    
    Conversely, suppose that there are two schedules~$\sigma_1, \sigma_2$ of makespan at most~$B$ and pairwise diversity~$d(\sigma_1,\sigma_2) \ge n$. 
    Since all weights are $1$, the maximum possible diversity is~$n$, so every job must be assigned to different machines in the two schedules. 
    Each job has only two eligible machines,~$m^\star$ and~$d_j$, and therefore exactly one of the two schedules assigns~$j$ to~$m^\star$. 
    Let~$S=\{j\in J \mid \sigma_1(j)=m^\star\}$.
    From the two makespan bounds we get
    \begin{align*}
        \sum_{j\in S}a_j \;\le\; B
        \qquad\text{and}\qquad
        \sum_{j\notin S}a_j \;\le\; B \enspace  .
    \end{align*}
    Since the total size is~$2B$, both inequalities are tight and~$S$ is a feasible partition.
    
    Finally, finding two distinct schedules of makespan at most $T$ is trivial in the reduction. 
    Choose any job $j_0$ (which has $p_{j_0} = a_{j_0} \le B = T$ by assumption) and create two distinct schedules in the following way:
    Assign all jobs to their private dummy machines in the first schedule, and assign only $j_0$ to $m^\star$ in the second schedule while assigning all other jobs to their private dummy machines. 
    Both schedules have makespan at most $T$.
\end{proof}

Note that the reduction shows hardness even in the graph-balancing machine model~\cite{EbenlendrKS2014} (where each eligible set is represented by an edge between two machines), on stars in particular.

\section{Max-sum diversity proofs}
\label{app:max-sum}

\subsection{Minimization \texorpdfstring{$k$}{k}-sided placement}
\label{subsec:min-k-sided-placement}
For completeness, we recall the minimization $k$-sided placement problem (\minkSP)~\cite{KorupoluMRT2015}. There are $k$ sides $S_1, \dots, S_k$, each consisting of capacitated nodes. Each job must be assigned to one $k$-tuple $e = (u_1, \dots, u_k)$, where $u_a \in S_a$ for every side~$a \in [k]$. For each job~$j$ and tuple~$e$, the input specifies a cost~$f_j(e)$ and a demand vector~$d_j(e)\in\mathbb{R}_{\ge 0}^k$. If job~$j$ is assigned to $e= (u_1, \dots, u_k)$, then it incurs cost~$f_j(e)$ and consumes the $a^{\textnormal{th}}$ coordinate of~$d_j(e)$ from the capacity of node~$u_a$, for every~$a \in [k]$. The goal is to assign every job to one tuple so as to minimize the total cost, while respecting all node capacities. 

In our application, side~$a$ corresponds to schedule~$a$, the nodes on each side are copies of the machines, and job~$j$ may only use tuples $(i_1,\dots,i_k)\in\Gamma(j)^k$. Assigning job~$j$ to tuple~$\mathbf{i}=(i_1,\dots,i_k)$ consumes demand~$p_{i_a j}$ on machine~$i_a$ of side~$a$. The cost of the tuple is 
\begin{align*}
    f_j(\mathbf{i})  =  w_j\sum_{i\in M}\binom{\mu_i^j(\mathbf{i})}{2}  ,
\end{align*}
where $\mu_i^j(\mathbf{i})=|\{a\in[k]\mid i_a=i\}|$.

\ktupleresult*
\begin{proof}
    Let $Q_0=\binom{k}{2}\sum_{j\in J} w_j$. 
    Note that any integral tuple assignment has total diversity exactly~$Q_0$ minus its total penalty~$\sum_{j\in J} \phi^j_{\mathbf{i}(j)}$. 
    Any~$k$ schedules of makespan at most~$T$ define a feasible $k$-sided placement of capacity~$T$, and conversely.
    In particular, an optimal solution attaining~$\OPT_k^\Sigma(T)$ has penalty~$Q_0 - \OPT_k^\Sigma(T)$, so the minimum penalty at capacity~$T$ is at most~$Q_0 - \OPT_k^\Sigma(T)$.
   
    Applying the $(k+1)$-capacity algorithm for \minkSP\ yields an integral tuple assignment whose penalty is at most this minimum penalty and whose capacities are at most~$(k+1)T$. 
    Translating the tuple assignment back into~$k$ schedules gives makespan at most~$(k+1)T$  and total diversity at least~$\OPT_k^\Sigma(T)$.
\end{proof}

\paragraph*{The multiplicity linear program.}

\multiplicityrelaxation*

\begin{proof}
    Let $\pi_1,\ldots,\pi_k$ be any~$k$ schedules of makespan at most~$T$ and set $\mu_{ij}=|\{a\in[k]\mid \pi_a(j)=i\}|$. 
    Define~$y_{ijs}=1$ for~$1 \le s \le \mu_{ij}$ and~$y_{ijs}=0$ otherwise. 
    We check that~$y$ is feasible for~\eqref{lp:multiplicity}.
    For every job~$j\in J$,  $\sum_{i,s} y_{ijs} = \sum_{i\in \Gamma(j)} \mu_{ij}=k$.
    Furthermore, for every machine~$i\in M$,
    \begin{align*}
        \sum_{j\in J :\, i\in \Gamma(j)} \sum_{s\in [k]} p_{ij} y_{ijs}
        \;=\;
        \sum_{j \in J :\, i\in \Gamma(j)} p_{ij}\mu_{ij}
        \;=\;
        \sum_{a\in [k]}\sum_{j \in J:\, \pi_a(j)=i}p_{ij}
        \;\le\; kT \enspace  .
    \end{align*}
    The last inequality holds because each schedule has load at most~$T$ on machine~$i$.
    Since~$y$ places copies in the lowest ranks, its objective value is
    \begin{align*}
        \sum_{j\in J} w_j \sum_{i\in \Gamma(j)} \sum_{s=1}^{\mu_{ij}}(s-1)
        \;=\; \sum_{j\in J} w_j \sum_{i\in \Gamma(j)} \binom{\mu_{ij}}{2}
        \;=\; \Phi(\boldsymbol \mu) \enspace  .
    \end{align*}
    Taking $\pi_1,\dots,\pi_k$ to be an optimal family for max-sum diversity with makespan at most $T$ implies $Z_{\mathrm{mult}}\le \Phi(\boldsymbol \mu)=Q_0-\OPT_k^\Sigma(T)$.
\end{proof}

\paragraph*{Rounding the multiplicity linear program.}

\multiplicityrounding*

\begin{proof}
  We prove this using iterative relaxation. 
  The cost of variable $y_{ijs}$ is $(s-1) w_j$.
  Moreover, the processing requirement of variable~$y_{ijs}$ is~$p_{ij} \leq p_{\max}$.

  During the algorithm, each variable is either already fixed to~$0$ or~$1$, or is still present.
  A job~$j$ is \emph{active} if it still has present variables.
    Furthermore, a machine~$i$ is \emph{active} if its load constraint has not been dropped.
    Denote the active jobs by~$J'$ and the active machines by~$M'$.
    For an active job~$j\in J'$, let~$b_j$ be the number of remaining ranked copies that must still be selected for job $j$. 
    Initially, $b_j=k$.
    For an active machine~$i\in M'$, let~$T_i^{\mathrm{res}}$ be its residual capacity.
    Initially, $T_i^{\mathrm{res}} = k T$.
    We write~$E$ for the set of tuples~$(i,j,s)$ such that the variable~$y_{ijs}$ is still present.
    Moreover, we denote by~$\delta(j)$ the set of tuples with present variables containing job~$j$ and by~$\delta(i)$ those containing machine~$i$.

    At every iteration, we compute an optimal extreme-point solution~$y$ of the following residual LP:
    \begin{alignat*}{3}
        \text{minimize} \qquad \sum_{(i,j,s)\in E} (s-1) &w_j y_{ijs} \\[-0.5mm]
        \text{subject to} \qquad
            \sum_{(i,j,s)\in\delta(j)} y_{ijs} \;&=\; b_j
            \qquad &&\forall j\in J'  ,\\[-0.5mm]
            \sum_{(i,j,s)\in\delta(i)} p_{ij} y_{ijs} \;&\le\; T_i^\mathrm{res}
            \qquad &&\forall i \in M'  ,\\[-0.5mm]
            0 \;\le\; y_{ijs} \;&\le\; 1 &&\forall (i,j,s) \in E  \enspace . 
    \end{alignat*}
    We apply the following rules to iteratively relax the LP.
    If~$y_{ijs} = 0$, then we remove the variable.
    If~$y_{ijs} = 1$, then we fix this variable to~$1$ (removing $y_{ijs}$ from the residual LP), decrease~$b_j$ by one, and if machine~$i$ is active, then decrease~$T_i^{\textrm{res}}$ by~$p_{ij}$.
    Once~$b_j = 0$, all remaining variables of job~$j$ are fixed to~$0$, and job~$j$ becomes inactive.
    It remains to specify what to do when every remaining variable is strictly fractional.

    Suppose that every present variable satisfies~$0<y_{ijs}<1$.
    Since no variable bound is tight, the extreme point is determined by $|E|$ linearly independent tight constraints among the active job equalities and the tight active machine constraints.
    In particular, this implies~$|E| \leq |J'| + |M'|$.
    Each active job~$j$ satisfies~$\sum_{(i,j,s) \in \delta(j)} y_{ijs} = b_j \geq 1$, hence
    \begin{align*}
        \sum_{(i,j,s)\in E} y_{ijs} 
        \;=\; \sum_{j\in J'} b_j 
        \;\ge\; |J'|  \enspace .
    \end{align*}
    Combining these two inequalities gives $\sum_{(i,j,s)\in E} (1-y_{ijs}) \le |M'|$.
    For an active machine~$i$, we define~$\Delta_i(y) = |\delta(i)| - \sum_{(i,j,s)\in \delta(i)} y_{ijs} = \sum_{(i,j,s)\in \delta(i)}(1-y_{ijs})$.
    Summing over active machine constraints implies
    \begin{align*}
        \sum_{i\in M'}\Delta_i(y)
        \;\le\; \sum_{(i,j,s)\in E} (1-y_{ijs})
        \;\le\; |M'|  \enspace .
    \end{align*}
    Moreover, $M'$ is nonempty.
    Indeed, if $M'=\emptyset$, then $|E| \le |J'|$. 
    Since all variables are strictly fractional, we have $\sum_{(i,j,s)\in E} y_{ijs} < |E|$, while $ \sum_{(i,j,s)\in E} y_{ijs} = \sum_{j\in J'} b_j \ge |J'|$, a contradiction. 
    Thus some active machine $i\in M'$ satisfies $\Delta_i(y)\le 1$.
    We drop this machine constraint.
    The algorithm terminates because every iteration either fixes at least one variable or drops one machine constraint.

    We now prove the load guarantee.
    Consider a machine~$i\in M$.
    If its constraint is never dropped, then its final aggregate load is at most its initial capacity~$kT$.
    Otherwise, consider the moment its constraint is dropped and let~$S = \delta(i)$ be its present variables then.
    Let~$T_i^\mathrm{fix}$ be the load already fixed on machine~$i$, so that~$T_i^\mathrm{fix} + T_i^\mathrm{res} = kT$, and recall that~$\sum_{(i,j,s)\in S} p_{ij}y_{ijs} \leq T_i^\mathrm{res}$ and~$\Delta_i(y) = \sum_{(i,j,s) \in S} (1-y_{ijs}) \le 1$.
    Let~$Y_{ijs} \in \{0,1\}$ be the final value of each variable.
    Then 
    \begin{align*}
        \sum_{(i,j,s)\in S} p_{ij} Y_{ijs} 
        \;&\le\; \sum_{(i,j,s)\in S}p_{ij} y_{ijs} + \sum_{(i,j,s)\in S} p_{ij} (Y_{ijs} - y_{ijs})^+ \\
        \;&\le\; T_i^\mathrm{res} + p_{\max}\sum_{(i,j,s)\in S} (Y_{ijs} - y_{ijs})^+ 
        \;\le\; T_i^\mathrm{res} + p_{\max}\sum_{(i,j,s)\in S} (1 - y_{ijs}) \\
        \;&=\; T_i^\mathrm{res} + p_{\max} \Delta_i(y)
        \;\le\; T_i^\mathrm{res} + p_{\max} \enspace .
    \end{align*}
    Hence the final load on machine~$i$ is at most~$T_i^\mathrm{fix} + T_i^\mathrm{res} + p_{\max} = kT + p_{\max}$.

    Finally, we maintain the invariant that the fixed cost plus the residual LP optimum is at most~$Z_{\mathrm{mult}}$.
    Fixing a variable that is integral in the current optimal extreme point preserves this invariant, and dropping a machine constraint can only decrease the optimum value of the residual LP. 
    Therefore, the final integral solution~$Y$ has cost at most~$Z_{\mathrm{mult}}$.
\end{proof}

\section{Max-min diversity proofs}
\label{app:max-min}

\subsection{Near-optimal max-min diversity}
\label{app:near-optimal-max-min}

\paragraph*{Feasibility after guessing.}

\feas*
\begin{proof}
The argument has two parts: construct a feasible assignment that performs
well for the weights given by the current diversity shortfalls, then
use it to decrease their sum of squares unless every target is met.

\emph{Light jobs suffice.}
Let \(L=\{j\in J\setminus G \mid w_j\le Q/h\}\).
An unguessed job outside \(L\) cannot contribute to an active pair
in the reference family.
If it did, the \(h\) largest contributors
to that pair would each have weight greater than \(Q/h\), making
\(Q_{a,b}^{\mathrm{fix}}>Q\), contrary to the pair being active.
If there were fewer than \(h\) contributors, all would have been guessed.
Subtracting the guessed contribution from the reference diversity gives
\begin{equation*}
 \sum_{j\in L}w_j\mathbf{1}[\pi_a(j)\ne\pi_b(j)]
 \ge Q-Q_{a,b}^{\mathrm{fix}}=Q_{a,b}^{\mathrm{res}}
 \qquad\text{for every active pair }(a,b) \enspace .
\end{equation*}

\emph{Measure the unmet targets.}
Let \(\mathcal P\) be the set of nonnegative points satisfying only
the job and load constraints of \eqref{lp:res}.
It is nonempty: for each unguessed job $j$ choose the constant tuple
\((\pi_1(j),\ldots,\pi_1(j))\), setting its variable to one and the
others to zero.
This tuple is in \(A_j\): its machine is eligible
because \(p_{\pi_1(j),j}\le T\), it uses one machine, and every diversity
coefficient is zero.
The job equalities hold, and the charged load on each machine is at most its load in \(\pi_1\), hence, at most \(T\le tT\).

For every pair $(a, b)$, denote its target and supplied residual diversity by
\begin{equation*}
 b_{a,b}=(1-\varepsilon_0)Q_{a,b}^{\mathrm{res}},\qquad
 d_{a,b}(z)=\sum_{j\notin G}\sum_{\mathbf{i}\in A_j}
                  q_{\mathbf{i},j}^{a,b}z_{\mathbf{i}}^j \enspace .
\end{equation*}
Choose \(z\in\mathcal P\) minimizing the sum of squared shortfalls
\begin{equation*}
 \Phi(z)=\sum_{a<b}\bigl((b_{a,b}-d_{a,b}(z))^+\bigr)^2 \enspace .
\end{equation*}
The minimum exists: linear equalities and non-strict inequalities are
preserved under limits, so~$\mathcal P$ is closed; nonnegativity
and the job equalities imply that every variable lies in \([0,1]\),
so it is bounded.
In finite dimensions these properties imply compactness.
Since $\mathcal P$ is nonempty and~$\Phi$ is continuous (a finite sum of squares of positive parts of linear functions), the
extreme value theorem guarantees an attained minimum.

If \(\Phi(z)=0\), the lemma follows. Otherwise set
\begin{equation*}
 \lambda_{a,b}=(b_{a,b}-d_{a,b}(z))^+,\qquad
 W_\lambda=\sum_{a<b}\lambda_{a,b}>0 \enspace .
\end{equation*}
These weights measure shortfalls of the whole point \(z\), independently
of any job, and remain fixed below.
Inactive pairs have
\(b_{a,b}=0\) and, hence, \(\lambda_{a,b}=0\); active pairs whose targets
are met also have weight zero.
Weighting by \(\lambda\) emphasizes
the unmet pairs in proportion to how strongly an increase in their
diversity reduces \(\Phi\).
We construct \(y\in\mathcal P\) satisfying
\begin{equation}\label{eq:weighted-witness}
 \sum_{a<b}\lambda_{a,b}d_{a,b}(y)
 \ge\sum_{a<b}\lambda_{a,b}b_{a,b} \enspace .
\end{equation}
The point \(y\) need not meet every target individually; this summed
weighted guarantee suffices to decrease \(\Phi\), as shown at the end.

\emph{Construct an assignment for these weights.}
We seek an assignment $y$ with large weighted diversity,
giving greater importance to pairs with larger shortfalls
$\lambda_{a,b}$. Exchanging finite sums gives
\begin{align*}
	\sum_{a<b}\lambda_{a,b}d_{a,b}(y)
	=\sum_{a<b}\sum_{j\notin G}\sum_{\mathbf{i}\in A_j}
	\lambda_{a,b}y_{\mathbf{i}}^j w_j\mathbf{1}[i_a\ne i_b]
=\sum_{j\notin G}\sum_{\mathbf{i}\in A_j}y_{\mathbf{i}}^j w_j
	\sum_{a<b}\lambda_{a,b}\mathbf{1}[i_a\ne i_b] \enspace .
\end{align*}
For an integral point $y$, exactly one tuple per job has variable value one, and all others have value zero.
Assigning job $j$ to different machines in schedules $a$ and $b$ therefore contributes $w_j\lambda_{a,b}$ to the weighted objective. 
With only a limited number of distinct machines available for each job, we seek tuples separating pairs of large total weight.

Consider the complete graph on schedule indices $[k]$, with edge weight $\lambda_{a,b}$ on pair $(a,b)$.
For a light job $j$ with $\ell$ distinct eligible machines available, let each color correspond to one of these machines.
Coloring vertex $a$ with the color of machine $i$ means assigning job $j$ to machine $i$ in schedule $a$.
Thus the job's weighted contribution is $w_j$ times the total weight of edges whose endpoints have different colors.
For $\ell\ge1$, define the maximum such \emph{cut weight}, achievable by $\ell$ colors, as
\begin{equation*}
g_\ell=\max_{c:[k]\to[\ell]}
\sum_{a<b}\lambda_{a,b}\mathbf{1}[c(a)\ne c(b)] \enspace .
\end{equation*}
A coloring attaining $g_\ell$ gives weighted contribution $w_jg_\ell$.
The edge weights are the same for every job,
and
\begin{equation*}
g_1=0,\qquad g_{\ell+1}\ge g_\ell,\qquad
g_\ell=W_\lambda\quad\text{for }\ell\ge k \enspace .
\end{equation*}

We choose the available machines by selecting for all jobs one common collection of~$t$ reference schedules among $\pi_1,\ldots, \pi_k$.
Restricting each job to machines used by these schedules will allow us to bound each machine's LP load by $tT$.
At each stage, define
\begin{equation*}
S_j=\{\pi_a(j) \mid \pi_a\text{ has been selected}\},
\qquad \ell_j=|S_j|,
\qquad F=\sum_{j\in L}w_jg_{\ell_j} \enspace .
\end{equation*}
Here $S_j$ is the set of available machines for job $j$, and $F$ is the total weighted diversity attainable by choosing a coloring separately for each light job using its machines in $S_j$.

Initially select $\pi_1$, so that $\ell_j=1$ and $F=0$.
Adding $\pi_a$ increases $\ell_j$ by one exactly when $\pi_a(j)\notin S_j$.
Its gain in $F$ is therefore
\begin{equation*}
\Delta_a
=\sum_{j\in L:\,\pi_a(j)\notin S_j}
w_j\bigl(g_{\ell_j+1}-g_{\ell_j}\bigr) \enspace .
\end{equation*}
Greedily add $t-1$ reference schedules, each maximizing $\Delta_a$, and update the sets $S_j$ after each addition.
Thus each step chooses one schedule based on its combined gain over all light jobs.
Repetitions are allowed; selecting an already chosen schedule gives zero gain.
We will show that the resulting $F$ is large enough to satisfy \eqref{eq:weighted-witness}.

For each $j\in L$, choose a coloring attaining $g_{\ell_j}$ and map its distinct colors to distinct machines in $S_j$.
The machine corresponding to the color of vertex $a$ becomes the $a^{\textnormal{th}}$ entry of the tuple for job~$j$.
For $j\notin G\cup L$, use the constant tuple
$(\pi_1(j),\ldots,\pi_1(j))$.
Denote the chosen tuple by $\mathbf{i}_j$, and set its variable in $y$ to one and all other variables for that job to zero.

Each tuple uses at most $t$ eligible machines. Among unguessed jobs, its diversity coefficients are zero outside $L$, and at most $w_j\le Q/h$ inside $L$.
Thus every chosen tuple belongs to $A_j$, and the job constraints hold.

To verify the load constraints, list the selected schedules as $\rho_1,\ldots,\rho_t$, including repetitions, with $\rho_1=\pi_1$. Every used job--machine pair occurs in
at least one selected schedule: for light jobs this follows from the definition of $S_j$, and for the other unguessed jobs from their constant-$\pi_1(j)$ tuples.
Hence
\begin{equation*}
\mathbf{1}[i\in S(\mathbf{i}_j)]
\le\sum_{s=1}^t\mathbf{1}[\rho_s(j)=i] \enspace .
\end{equation*}
The left side counts machine $i$ once if the tuple uses it; the right side counts the selected schedules assigning job $j$ to $i$. 
Multiplying by $p_{ij}\ge0$ and summing over unguessed jobs verifies the load constraint of
\eqref{lp:res}:
\begin{align*}
	\sum_{j\notin G}\sum_{\mathbf{i}\in A_j:\,i\in S(\mathbf{i})}
	p_{ij}y_{\mathbf{i}}^j
	&=\sum_{j\notin G}p_{ij}\mathbf{1}[i\in S(\mathbf{i}_j)]\\
	&\le\sum_{j\notin G}p_{ij}
	\left(\sum_{s=1}^t\mathbf{1}[\rho_s(j)=i]\right)\\
	&=\sum_{s=1}^t\sum_{j\notin G:\,\rho_s(j)=i}p_{ij}\\
	&\le\sum_{s=1}^t T=tT \enspace .
\end{align*}
The first equality uses that $y$ selects exactly $\mathbf{i}_j$ for each job $j$.
The final inequality follows because each inner sum is at most the load of all jobs on machine $i$ in $\rho_s$, hence, at most $T$.
Consequently, $y\in\mathcal P$.
The light jobs contribute $w_jg_{\ell_j}$ each, while the other unguessed jobs
contribute zero, so
\begin{equation*}
\sum_{a<b}\lambda_{a,b}d_{a,b}(y)
=\sum_{j\in L}w_jg_{\ell_j}=F \enspace .
\end{equation*}

\emph{Bound the greedy gain.}
Let $F_s$ denote the value of $F$ after selecting $s$ reference
schedules.
We compare $F_s$ with the reference family's weighted diversity on the light jobs. For $j\in L$, define
\begin{equation*}
D_j=\sum_{a<b}\lambda_{a,b}\mathbf{1}[\pi_a(j)\ne\pi_b(j)],
\qquad D=\sum_{j\in L}w_jD_j \enspace .
\end{equation*}
Here $D_j$ is the cut weight of the reference coloring for job $j$, and $w_jD_j$ is its contribution to weighted diversity.
In contrast, $g_{\ell_j}$ is the maximum cut weight attainable with at most $\ell_j$ colors.
The light-job bound gives $D\ge\sum_{a<b}\lambda_{a,b}Q_{a,b}^{\mathrm{res}}>0$, since every positive-weight pair is active.
Our goal is to show that, after $t$ selections,
$F_t\ge(1-\varepsilon_0)D$.
This implies
\begin{equation*}
\sum_{a<b}\lambda_{a,b}d_{a,b}(y)
=F_t\ge(1-\varepsilon_0)D
\ge\sum_{a<b}\lambda_{a,b}(1-\varepsilon_0)Q_{a,b}^{\mathrm{res}}
=\sum_{a<b}\lambda_{a,b}b_{a,b},
\end{equation*}
as required by \eqref{eq:weighted-witness}.

Fix a step $1\le s<t$. In the following argument, $S_j$, $\ell_j$, $\Delta_a$, and the quantities $u_j$ and $e_j$ defined below refer to their values after $s$ selections.
In particular, $F_s=\sum_{j\in L}w_jg_{\ell_j}$.
To bound the gain from adding another reference schedule, we measure the weighted degree at vertices whose reference machines are not yet in $S_j$.
Let
\begin{equation*}
\deg(a)=\sum_{b<a}\lambda_{b,a}+\sum_{b>a}\lambda_{a,b},
\qquad \sum_a\deg(a)=2W_\lambda,
\qquad
u_j=\frac{\sum_{a:\pi_a(j)\notin S_j}\deg(a)}{2W_\lambda} \enspace .
\end{equation*}
Thus $u_j$ is the fraction of total weighted degree at vertices whose reference colors must change when restricting the coloring to machines in $S_j$.

To upper-bound the remaining deficit $D-F_s$, we lower-bound the increase $F_{s+1}-F_s$ at each greedy step.
Since the coefficients $\deg(a)/(2W_\lambda)$ are nonnegative and sum to one, the largest gain is at least their weighted average:
\begin{align*}
	\max_a\Delta_a
	&\ge \sum_a\frac{\deg(a)}{2W_\lambda}\Delta_a\\
	&=\sum_a\frac{\deg(a)}{2W_\lambda}
	\sum_{j\in L:\,\pi_a(j)\notin S_j}
	w_j(g_{\ell_j+1}-g_{\ell_j})\\
	&=\sum_{j\in L}w_j(g_{\ell_j+1}-g_{\ell_j})
	\left(\sum_{a:\,\pi_a(j)\notin S_j}
	\frac{\deg(a)}{2W_\lambda}\right)\\
	&=\sum_{j\in L}w_ju_j(g_{\ell_j+1}-g_{\ell_j}) \enspace .
\end{align*}

Let $e_j=(D_j-g_{\ell_j})^+$ measure the positive deficit relative to the reference coloring for job $j$.
We will show that, for $D_j>0$,
\begin{equation*}
u_j(g_{\ell_j+1}-g_{\ell_j})\ge\frac{e_j^2}{4D_j} \enspace .
\end{equation*}
For this, we establish two recoloring bounds:
\begin{equation}
\label{eq:cuts}
	g_{\ell+1}-g_\ell\ge\frac{W_\lambda-g_\ell}{2\ell},
	\qquad
	g_{\ell_j}\ge D_j-\frac{2W_\lambda u_j}{\ell_j} \enspace .
\end{equation}
For the first, an optimal $\ell$-coloring has total internal edge weight $W_\lambda-g_\ell$.
One color class has internal weight at least $(W_\lambda-g_\ell)/\ell$, allowing empty classes.
Split that class into two parts maximizing the weight between them.
Each vertex has at least as much incident weight across the split as within its part, since otherwise moving it improves the split.
Summing over vertices counts each edge twice, so at least half the class's internal weight crosses the split.
Giving one part a new color cuts this weight without losing any previously cut edge, proving the first bound.

For the second bound, start with the reference colors $\pi_a(j)$.
Keep vertices whose colors belong to $S_j$ fixed, and recolor the others using colors in $S_j$ to maximize the cut weight.
For each recolored vertex $a$, the incident weights to the $\ell_j$ color classes sum to $\deg(a)$, so some class has
incident weight at most $\deg(a)/\ell_j$. If its current color class had greater incident weight, moving $a$ to a least-weight class would increase the cut weight, contradicting optimality.
Thus each recolored vertex has at most $\deg(a)/\ell_j$ incident weight to vertices of its own color.
Every originally cut edge that is lost must touch a recolored vertex, since two fixed endpoints retain their different colors.
At that recolored vertex, the lost edge now joins vertices of the same color. 
Summing the preceding bounds therefore gives
\begin{equation*}
D_j-g_{\ell_j}
\le \frac{1}{\ell_j}
\sum_{a:\pi_a(j)\notin S_j}\deg(a)
= \frac{2W_\lambda u_j}{\ell_j} \enspace .
\end{equation*}
The reference cut weight is $D_j$, and newly cut edges have nonnegative weight.
As the resulting coloring uses at most $\ell_j$ colors, it proves the second bound.
The second bound implies~$u_j\ge\ell_j e_j/(2W_\lambda)$.
Combining the bounds in \eqref{eq:cuts} gives, for $D_j>0$,
\begin{equation*}
u_j(g_{\ell_j+1}-g_{\ell_j})
\ge\frac{e_j(W_\lambda-g_{\ell_j})}{4W_\lambda}
\ge\frac{e_j^2}{4D_j} \enspace .
\end{equation*}
For the last inequality, when $e_j>0$, use $D_j\le W_\lambda$ to get
$(W_\lambda-g_{\ell_j})/W_\lambda
\ge(D_j-g_{\ell_j})/D_j=e_j/D_j$.
When $e_j=0$, the bound is immediate; $D_j=0$ also implies $e_j=0$.

Consequently, whenever $F_s<D$,
\begin{align}
\label{eq:gain}
	F_{s+1}-F_s
	&=\max_a\Delta_a
	\ge\sum_{j\in L}w_ju_j(g_{\ell_j+1}-g_{\ell_j})
	\ge\frac14\sum_{j\in L:\,D_j>0}\frac{w_je_j^2}{D_j}
	\notag\\
	&=\frac1{4D}
	\left(\sum_{j\in L:\,D_j>0}w_jD_j\right)
	\left(\sum_{j\in L:\,D_j>0}\frac{w_je_j^2}{D_j}\right)
	\notag\\
	&\ge\frac1{4D}
	\left(\sum_{j\in L:\,D_j>0}
	\sqrt{w_jD_j}\,
	\sqrt{\frac{w_je_j^2}{D_j}}\right)^2
	=\frac{(\sum_{j\in L}w_je_j)^2}{4D}
	\ge\frac{(D-F_s)^2}{4D} \enspace .
\end{align}
The third inequality follows from the Cauchy--Schwarz inequality
$\left(\sum_j a_jb_j\right)^2
\le \left(\sum_j a_j^2\right)\left(\sum_j b_j^2\right)$.
The last inequality uses
$\sum_jw_je_j\ge\sum_jw_j(D_j-g_{\ell_j})=D-F_s>0$.

Initially $F_1=0$, and $F_s$ never decreases.
While $\delta_s=(D-F_s)/D>0$, \eqref{eq:gain} gives
$\delta_{s+1}\le\delta_s-\delta_s^2/4$.
If the new deficit remains positive, then
\begin{equation*}
\frac1{\delta_{s+1}}
\ge\frac1{\delta_s(1-\delta_s/4)}
=\frac1{\delta_s}+\frac1{4(1-\delta_s/4)}
\ge\frac1{\delta_s}+\frac14 \enspace .
\end{equation*}
Since initially $\delta_1=1$, after $t-1$ additions, 
we have $F_t\ge D$ or $\frac1{\delta_t}\ge1+\frac{t-1}{4}=\frac{t+3}{4}.$
In the latter case, $\frac{D-F_t}{D} = \delta_t\le\frac4{t+3} \le \varepsilon_0 .$
Thus, in either case, $F_t\ge(1 -\varepsilon_0) D\ge\sum_{a<b}\lambda_{a,b}b_{a,b},$ which proves~\eqref{eq:weighted-witness}.

\emph{Average to reduce the shortfall.}
Set \(v_{a,b}=d_{a,b}(y)-d_{a,b}(z)\).
By
\eqref{eq:weighted-witness},
\begin{equation*}
 \sum_{a<b}\lambda_{a,b}v_{a,b}
 \ge\sum_{a<b}\lambda_{a,b}(b_{a,b}-d_{a,b}(z))
 =\sum_{a<b}\lambda_{a,b}^2=\Phi(z)>0 \enspace .
\end{equation*}
For \(0<\theta\le1\), form the convex combination
\(z_\theta=(1-\theta)z+\theta y\), meaning
\((z_\theta)_{\mathbf{i}}^j=(1-\theta)z_{\mathbf{i}}^j+\theta y_{\mathbf{i}}^j\)
coordinatewise.
It lies in \(\mathcal P\) because linear constraints
and nonnegativity are preserved by convex combinations, and
\(d_{a,b}(z_\theta)=d_{a,b}(z)+\theta v_{a,b}\).
The inequality \(((x-s)^+)^2\le(x^+)^2-2sx^++s^2\), applied with \(x=b_{a,b}-d_{a,b}(z)\) and \(s=\theta v_{a,b}\), gives
\begin{equation*}
 \Phi(z_\theta)
 \le\Phi(z)-2\theta\sum_{a<b}\lambda_{a,b}v_{a,b}
              +\theta^2\sum_{a<b}v_{a,b}^2
 \le\Phi(z)-2\theta\Phi(z)+\theta^2 \sum_{a<b}v_{a,b}^2,
\end{equation*}
using \eqref{eq:weighted-witness} in the last step.
Since $\sum_{a<b}\lambda_{a,b}v_{a,b}>0$, we have that $\sum_{a<b}v_{a,b}^2>0$.
By taking $\theta=\min\{1,\Phi(z)/\sum_{a<b}v_{a,b}^2\}>0$, we get $\theta \sum_{a<b}v_{a,b}^2 \le\Phi(z)$.
Then $\Phi(z_\theta)\le(1-\theta)\Phi(z)<\Phi(z)$, contradicting minimality.
Consequently $\Phi(z)=0$, so every diversity target is met, and $z$ satisfies \eqref{lp:res}.
\end{proof}

\paragraph*{Rounding the residual LP.}

\roundingrevised*
\begin{proof}
Assume that \eqref{lp:res} is feasible.
We use iterative rounding as long as there are unassigned unguessed jobs, maintaining feasibility throughout.
Each iteration has the following steps:
First, we remove every diversity constraint whose current residual right-hand side is non-positive.
Second, we compute an extreme-point solution of the resulting LP.
Third, we fix variables that were assigned integrally in the extreme-point solution:
If $z_\mathbf{i}^j = 0$, then we delete the variable from the LP.
If $z_\mathbf{i}^j = 1$, then we commit to assigning unguessed job~$j$ to tuple~$\mathbf{i}$ by deleting~$j$ with all its variables and updating the load constraints and the diversity constraints accordingly.
Specifically, for every machine~$i \in S(\mathbf{i})$, we subtract $p_{i j}$ from the remaining load budget of the machine, and we subtract $q_{\mathbf{i}, j}^{a, b}$ from the remaining diversity target of each active pair~$(a, b)$.
We end the iteration if at least one variable is deleted.
Finally, if no variable was fixed, we drop either a load constraint or a diversity constraint:
To this end, define $\Delta_i(z) = \sum_{j, \mathbf{i} \colon i \in S(\mathbf{i})} p_{i j} (1 - z_\mathbf{i}^j)$ and $V_{a, b}(z) = \sum_{j, \mathbf{i}} q_{\mathbf{i}, j}^{a, b} z_\mathbf{i}^j$.
We drop the load constraint of some machine~$i$ if $\Delta_i(z) \leq 2 t T$.
Otherwise, we drop the diversity constraint of an active pair~$(a, b)$ if $V_{a, b}(z) \leq \varepsilon_0 Q$.
This completes one iteration.

We show that, in every iteration where no variable is fixed, one constraint is dropped necessarily.
Consider an iteration where the remaining variables $E \neq \emptyset$ are all assigned fractional values strictly between~$0$ and~$1$ in the extreme-point solution.
Consequently, none of the variable constraints $z_\mathbf{i}^j \geq 0$ are tight.
We use $J'$ to denote the set of remaining jobs, and we use $R'$ to denote the set of remaining load and diversity constraints.
Then, by the rank lemma, the number~$|E|$ of (non-zero) variables is equal to the maximum number of linearly independent tight constraints, which is at most~$|J'| + |R'|$.
Further, for each job~$j \in J'$, there are at least two variables, since $z_\mathbf{i}^j < 1$ for every tuple~$\mathbf{i}$.
Hence, $2 |J'| \leq |E|$.
Combined with $|E| \leq |J'| + |R'|$, this gives $|E| \leq 2 |R'|$.
As $E \neq \emptyset$, we then get that $R' \neq \emptyset$.
We now argue that one constraint from $R'$ will be dropped.
If $T=0$, every eligible processing time is zero, so any remaining load constraint has $\Delta_i(z)=0$ and can be dropped. If no load constraint remains, then $|J'|\leq |R'|\leq r$, and every remaining diversity constraint satisfies
\begin{equation*}
V_{a,b}(z)\leq \frac{Q}{h}\sum_{(j,\mathbf{i})\in E}z_{\mathbf{i}}^j =\frac{Q}{h}|J'|\leq\frac{Qr}{h}\leq\frac{\varepsilon_0 Q}{2} \enspace.
\end{equation*}
Since $R'\neq\emptyset$, a diversity constraint can then be dropped.
Henceforth assume $T>0$, and define $\beta = 2 t T / (\varepsilon_0 Q)$.
It suffices to show that $\Delta_i(z) \leq 2 t T$ for some machine~$i$ with remaining load constraint or that $\beta V_{a, b} (z) \leq 2 t T$ for some active pair~$(a, b)$ with remaining diversity constraint.
We show this using two bounds.
In the following sums, machines and pairs range only over constraints in~$R'$.

First, we need the bound
\begin{equation}
\label{eq:sum-drop-machine}
\sum_i \Delta_i(z) \leq T \sum_i \sum_{j, \mathbf{i} \colon i \in S(\mathbf{i})} (1 - z_\mathbf{i}^j) \leq t T \sum_{j, \mathbf{i}} (1 - z_\mathbf{i}^j) = t T (|E| - |J'|),
\end{equation}
where we used that $p_{i j} \leq T$, that $|S(\mathbf{i})| \leq t$, and that $\sum_{j, \mathbf{i}} z_\mathbf{i}^j = |J'|$.

Second, we need the bound
\begin{equation}
\label{eq:sum-drop-pair}
\beta \sum_{(a, b)} V_{a, b}(z) \leq \beta (Q / h) \sum_{(a, b)} \sum_{j, \mathbf{i}} z_\mathbf{i}^j \leq \beta (Q / h) r \sum_{j, \mathbf{i}} z_\mathbf{i}^j \leq t T \sum_{j, \mathbf{i}} z_\mathbf{i}^j = t T |J'|,
\end{equation}
where we used that $q_{\mathbf{i}, j}^{a, b} \leq Q / h$, that there are at most $r$ active pairs, that $\beta (Q / h) r \leq t T$, and that $\sum_{j, \mathbf{i}} z_\mathbf{i}^j = |J'|$.

Combining \eqref{eq:sum-drop-machine} and \eqref{eq:sum-drop-pair}, we obtain
\begin{equation}
\label{eq:sum-drop-machine-and-pair}
\sum_i \Delta_i(z) + \beta \sum_{(a, b)} V_{a, b}(z) \leq t T (|E| - |J'|) + t T |J'| = t T |E| \leq 2 t T |R'| \enspace .
\end{equation}
Since the two sums in \eqref{eq:sum-drop-machine-and-pair} have $|R'|$ terms in total, at least one of these terms must then lead to a constraint being dropped.

It remains to bound the machine loads and the diversity after rounding.
Consider a machine~$i$ when its load constraint is dropped.
Let $B_i$ be its charged load from previously fixed unguessed-job tuples, and let $E_i=\{(j,\mathbf{i})\in E \mid i\in S(\mathbf{i})\}$ at that moment.
The load constraint gives $B_i+\sum_{(j,\mathbf{i})\in E_i}p_{ij}z_{\mathbf{i}}^j\leq tT$.
Writing $\hat z_{\mathbf{i}}^j$ for each remaining variable's final integral value, we obtain
\begin{equation*}
B_i+\sum_{(j,\mathbf{i})\in E_i}p_{ij}\hat z_{\mathbf{i}}^j
\leq B_i+\sum_{(j,\mathbf{i})\in E_i}p_{ij}z_{\mathbf{i}}^j+\Delta_i(z)
\leq tT+2tT=3tT,
\end{equation*}
because $\hat z_{\mathbf{i}}^j-z_{\mathbf{i}}^j\leq1-z_{\mathbf{i}}^j$.
This charged load bounds the load in every schedule.
A machine whose load constraint is never dropped has charged load at most~$tT$.

To bound the diversity after rounding, first consider a constraint removed because its current residual right-hand side is non-positive.
The fixed unguessed jobs already meet its original target $(1-\varepsilon_0)Q_{a,b}^{\mathrm{res}}$.
Otherwise, when a diversity constraint is dropped, its remaining fractional diversity is at most~$\varepsilon_0 Q$.
The unguessed jobs fixed before that moment therefore contribute at least $(1-\varepsilon_0)Q_{a,b}^{\mathrm{res}}-\varepsilon_0 Q$.
A diversity constraint retained until all jobs are assigned is satisfied, and later assignments cannot decrease diversity.

Since every iteration deletes a variable or drops a load/diversity constraint, the number of iterations is at most~$n m^k + m + r$.
Thus, the overall run time is polynomial in the size of the LP.
\end{proof}

\paragraph*{Guessing the diversity target and the heavy contributors.}

\nearoptimalrevised*
\begin{proof}
We start by guessing the diversity target to use.
To this end, define $Q^*=\OPT_k^{\min}(T)$ and $\varepsilon_0=\varepsilon/3$.
First, compute a baseline family of schedules using \Cref{thm:max-sum-two}; every schedule has makespan at most~$2T$. Let $W=\sum_j w_j$. If $W=0$, return this baseline family. Otherwise, let $w_{\min}=\min\{w_j \mid w_j>0\}$.
Enumerate diversity targets $Q=W/(1+\varepsilon_0)^\ell$ for $\ell=0,1,\ldots$, including and stopping at the first value with $Q\leq w_{\min}$.
The baseline family covers the case where $Q^*=0$. For the analysis of $Q^*>0$, we have $w_{\min}\leq Q^*\leq W$, so some enumerated value~$\widehat{Q}$ satisfies $\widehat{Q}\leq Q^*<(1+\varepsilon_0)\widehat{Q}$.

For each enumerated value of~$Q$, we guess the heavy contributors.
Specifically, for every pair~$(a, b)$ with $1 \leq a < b \leq k$, we guess the $h = \lceil 2 r / \varepsilon_0 \rceil$ jobs with the largest positive contribution to the diversity of schedules $a$ and $b$, and we guess the $k$-tuple of machines to which each of these jobs is assigned in the $k$ schedules.
Should there be fewer than $h$ jobs that positively contribute to the diversity of $a$ and $b$, then we guess all of them, again together with their $k$-tuples.
We discard inconsistent guesses and guesses whose fixed jobs violate a makespan constraint, as described in \Cref{subsec:near-optimal-max-min-diversity}.
For each remaining guess, we set up \eqref{lp:res} and discard the guess if this LP is infeasible. Otherwise, we apply \Cref{lem:rounding-revised} and combine the rounded assignment with the guessed tuples.
We return a family of maximum actual max-min diversity among these candidates and the baseline family.

Among our guesses will be $\widehat{Q}$ and its correct heavy-contributor guess.
Then, by \Cref{lem:feas}, we know that \eqref{lp:res} is feasible.
From \Cref{lem:rounding-revised}, we thus obtain an assignment of the unguessed jobs with unguessed-job load at most $3 t T$ per machine, where $t = \lceil 4 / \varepsilon_0 \rceil$, and diversity at least $(1 - \varepsilon_0) Q_{a, b}^\mathrm{res} - \varepsilon_0 \widehat{Q}$ from unguessed jobs, for every active pair~$(a, b)$.
Adding these unguessed jobs to the heavy-contributor guess then gives makespan at most $3 t T + T = (3 t + 1) T = (3 \lceil 12 / \varepsilon \rceil + 1) T$.
For each active pair, the diversity is at least
\begin{equation*}
  (1 - \varepsilon_0) Q_{a, b}^\mathrm{res} - \varepsilon_0 \widehat{Q} + Q_{a, b}^\mathrm{fix} \geq (1 - \varepsilon_0) (\widehat{Q} - Q_{a, b}^\mathrm{fix}) - \varepsilon_0 \widehat{Q} + Q_{a, b}^\mathrm{fix} \geq (1 - 2 \varepsilon_0) \widehat{Q} \enspace .
\end{equation*}
Inactive pairs already receive at least~$\widehat{Q}$ from guessed jobs.
Since $(1 + \varepsilon_0) \widehat{Q} > Q^*$, this gives diversity at least $\frac{1 - 2 \varepsilon_0}{1 + \varepsilon_0} Q^* \geq (1 - 3 \varepsilon_0) Q^* = (1 - \varepsilon) Q^*$.

The run time is dominated by guessing the heavy contributors and their $k$-tuples:
As $h = \mathcal{O}(k^2 / \epsilon)$, the number of these guesses is in $(n m^k)^{\mathcal{O}(\binom{k}{2} h)} = (n m^k)^{\mathcal{O}(k^4 / \varepsilon)}$.
This gives the promised run time of $(n m^k)^{\mathcal{O}(k^4 / \varepsilon)} \cdot |I|^{O(1)}$.
\end{proof}

\subsection{Near-optimal max-min diversity with a \texorpdfstring{$2k+1+\varepsilon$}{2k + 1 + epsilon} makespan factor}
\label{app:near-optimal-max-min-small-k}

In this subsection, we give an alternative algorithm that achieves near-optimal max-min diversity with makespan at most~$(2k+1+\varepsilon)T$. For small fixed~$k$, this improves the dependence on~$\varepsilon$ compared with the~$\mathcal{O}(1/\varepsilon)$ makespan guarantee in~\Cref{subsec:near-optimal-max-min-diversity}.
The algorithm guesses a diversity target and the largest contributors to each pairwise distance, together with their assignments in all~$k$ schedules.
After fixing these assignments and removing large residual contributors, it rounds the residual labeled $k$-tuple LP by iterative rounding, using scaled diversity coefficients to control both the load increase and the diversity loss.

\begin{restatable}[Near-optimal max-min diversity]{theorem}{maxminktuple}
\label{thm:max-min-ktuple}
  Fix $k \in \mathbb{N}_{\ge 2}$ and $\varepsilon\in (0,1)$.
  There is an algorithm that, for every feasible makespan target $T$, computes schedules $\sigma_1,\dots,\sigma_k$ such that $\Cmax(\sigma_a)\le (2k+1+\varepsilon)T$, for every $a \in [k]$, and $d_{\min}(\sigma_1,\dots,\sigma_k) \ge  (1-\varepsilon)\OPT_k^{\min}(T)$. 
  The algorithm runs in time $(nm^k)^{\mathcal{O}(k^5/\varepsilon^2)}\cdot |I|^{O(1)}$, where $|I|$ is the input encoding length.
\end{restatable}

We start from the natural labeled $k$-tuple LP. 
For each job~$j$ and tuple $\mathbf{i} = (i_1,\ldots,i_k)\in\Gamma(j)^k$, we introduce a variable~$z^j_{\mathbf{i}}$ that indicates whether job~$j$ is placed on machine~$i_a$ in schedule~$a$, for all~$a\in [k]$.
Together, the tuples assigned to the jobs encode~$k$ labeled schedules.
For a pair $1\le a<b\le k$, let~$q^{a,b}_{\mathbf{i} ,j} = w_j \mathbf{1}[i_a\ne i_b]$; this is the contribution of job~$j$ to the diversity of pair~$(a,b)$ when the tuple~$\mathbf{i}$ places~$j$ on different machines in coordinates~$a$ and~$b$.
The key insight is that the diversity of pair~$(a,b)$ is the linear function~$\sum_{j\in J} \sum_{\mathbf{i}\in \Gamma(j)^k} q_{\mathbf{i},j}^{a,b} z_{\mathbf{i}}^j$.
Thus, requiring the minimum diversity of the~$k$ schedules to be at least~$Q$ can be expressed by~$\binom{k}{2}$ linear covering constraints.
Hence, we obtain the following labeled~$k$-tuple LP:
\begin{alignat}{3}
    \sum_{\mathbf i\in\Gamma(j)^k} z^j_{\mathbf i}
    \;&=\; 1
    \qquad&& \forall j\in J, \tag{Tuple-LP}\label{lp:initial-ktuple-maxmin}\\[-0.5mm]
    \sum_{j\in J} \sum_{\substack{\mathbf i\in\Gamma(j)^k:\, i_a=i}} p_{ij}z^j_{\mathbf i}
    \;&\le\; T
    \qquad&& \forall a\in[k],\ i\in M,\notag\\[-0.5mm]
    \sum_{j\in J} \sum_{\mathbf i\in\Gamma(j)^k} q^{a,b}_{\mathbf i,j}z^j_{\mathbf i}
    \;&\ge\; Q
    \qquad&& \forall 1\le a<b\le k,\notag\\[-0.5mm]
    z^j_{\mathbf i}
    \;&\ge\; 0
    \qquad&& \forall j\in J,\ \mathbf i\in\Gamma(j)^k \enspace .\notag
\end{alignat}

Fix $Q>0$ and schedules $\pi_1,\ldots,\pi_k$ of makespan at most $T$ such that $d(\pi_a,\pi_b)\ge Q$ for every $1\le a<b\le k$.
Let $\varepsilon_0\in(0,1)$, and set
\begin{align*}
    \gamma \;=\; \frac{\varepsilon_0}{2k} \qquad \text{and} \qquad
    h \;=\; \left\lceil \frac{2(1+\gamma)\binom{k}{2}}{\gamma\varepsilon_0} \right\rceil \enspace .
\end{align*}

\paragraph*{Guessing the heavy contributors.}

For every pair $(a,b)$, the correct guess selects the $h$ largest positive contributors to $d(\pi_a,\pi_b)$, or all positive contributors if there are fewer than $h$, together with their full $k$-tuples in $\pi_1,\ldots,\pi_k$.
Let $G$ be the set of guessed jobs.
For each $j\in G$, let $\tau(j)=(\tau_1(j),\ldots,\tau_k(j))$ be the guessed $k$-tuple, where $\tau_a(j)$ is the guessed machine of $j$ in schedule $a$.
For the correct guess, $\tau_a(j)=\pi_a(j)$ for every $j\in G$ and $a\in[k]$.
We define
\begin{align*}
    T^{\mathrm{fix}}_{a,i} 
    \;=\; \sum_{\substack{j\in G \\ \tau_a(j)=i}} p_{ij},
    \qquad T^{\mathrm{res}}_{a,i}
    \;=\; T-T^{\mathrm{fix}}_{a,i} ,
\end{align*}
and
\begin{align*}
    Q^{\mathrm{fix}}_{a,b} 
    \;=\; \sum_{j\in G} w_j\mathbf 1[\tau_a(j)\ne\tau_b(j)],
    \qquad Q^{\mathrm{res}}_{a,b}
    \;=\;\left(Q-Q^{\mathrm{fix}}_{a,b}\right)^+  \enspace .
\end{align*}
We say that pair~$(a,b)$ is \emph{active} when $Q^{\mathrm{res}}_{a,b}>0$. 
For every unguessed job~$j$, we start with all tuples in~$\Gamma(j)^k$ and delete every tuple~$\mathbf{i}$ with~$q^{a,b}_{\mathbf{i},j} > Q/h$ for some active pair~$(a,b)$. 
Let~$\mathcal{A}_j$ be the remaining tuples. 
The residual $k$-tuple LP is obtained from \eqref{lp:initial-ktuple-maxmin} by restricting to jobs~$J\setminus G$, replacing~$\Gamma(j)^k$ by $\mathcal A_j$, replacing the load targets by~$T^{\mathrm{res}}_{a,i}$, and imposing diversity targets~$Q^{\mathrm{res}}_{a,b}$ for all active pairs.

\begin{lemma}[Removing high-contribution pairs is safe]
    \label{lem:correct-guess-bounded}
    No tuple used by $\pi_1,\ldots,\pi_k$ is deleted. 
    Therefore, the residual $k$-tuple LP is feasible, and every remaining tuple satisfies $q^{a,b}_{\mathbf i,j}\le Q/h$ for every active pair~$(a,b)$.
\end{lemma}
\begin{proof}
    Suppose that the tuple of an unguessed job~$j$ used by $\pi_1,\ldots,\pi_k$ is deleted for an active pair~$(a,b)$. 
    Then~$q^{a,b}_{\mathbf i,j}>Q/h$. 
    If fewer than $h$ jobs contribute positively to~$d(\pi_a,\pi_b)$, all of them were guessed, including~$j$, a contradiction.
    Otherwise, the~$h$ guessed contributors for this pair each have contribution at least~$q^{a,b}_{\mathbf i,j}>Q/h$, so their total contribution is larger than~$Q$. 
    This gives~$Q^{\mathrm{fix}}_{a,b}>Q$, contradicting that~$(a,b)$ is active. 
    Thus no tuple used by $\pi_1,\ldots,\pi_k$ is deleted, and restricting these schedules to the unguessed jobs makes the LP feasible. 
    The bounded-coefficient property holds by construction.
\end{proof}

\paragraph*{Rounding the residual LP.}

For the following rounding analysis, we assume $T>0$.
We now round the feasible residual LP produced by the correct guess.
It has tuple sets~$\mathcal A_j$, load targets~$T^{\mathrm{res}}_{a,i}$, and active diversity targets~$Q^{\mathrm{res}}_{a,b}$. 
Furthermore, we define constants
\begin{align*}
    \beta=\frac{\gamma kT h}{\binom{k}{2}Q},\qquad \qquad \qquad
    \Lambda=(1+\gamma)kT  \enspace .
\end{align*}
The iterative relaxation keeps the true, unscaled load and diversity constraints. 
We only use scaling to measure the mass of diversity constraints. 
The scaled coefficient of~$z^j_{\mathbf{i}}$ in diversity constraint~$(a,b)$ is~$\beta q^{a,b}_{\mathbf{i},j}$.
The scaled mass of a constraint is the sum of its remaining coefficients, with diversity coefficients multiplied by~$\beta$ and load coefficients unchanged.

Since every tuple is eligible, a tuple variable has total load mass at most~$kT$. 
Since $q^{a,b}_{\mathbf{i},j}\le Q/h$ for every active pair, its total scaled diversity mass is at most~$\binom{k}{2}\beta\frac{Q}{h} = \gamma kT$.
Thus every variable has total scaled mass over all remaining load constraints and active diversity constraints at most~$\Lambda$.

\begin{lemma}[Small-mass constraint]
\label{lem:small-mass-constraint}
    If the current residual LP has at least one variable and an extreme point with no integral variable, then some remaining load constraint or active diversity constraint has remaining scaled coefficient mass at most $2\Lambda$.
\end{lemma}
\begin{proof}
  Let~$E$ be the set of remaining tuple variables, let~$J'$ be the set of active jobs, and~$R'$ be the set of remaining load constraints and active diversity constraints. 
  Since no variable bound is tight, the extreme point is determined by the active job equalities and the tight remaining load or active diversity constraints, so~$|E|\le |J'|+|R'|$. 
  Every active job has right-hand side~$1$, and all its remaining variables are non-integral, so it has at least two remaining tuple variables. 
  Hence~$|J'|\le |E|/2$, and thus~$|R'|\ge |E|/2$. 
  Since the total scaled mass over all constraints in~$R'$ is at most~$\Lambda|E|$, the pigeonhole principle implies that some constraint in~$R'$ has mass at most~$2\Lambda$.
\end{proof}

\begin{restatable}[Rounding after guessing]{lemma}{roundingafterguessing}
\label{lem:rounding-after-guessing}
  For the correct heavy-contributor guess, the residual $k$-tuple LP is feasible and can be rounded in polynomial time for fixed~$k$ and~$\varepsilon_0$ to schedules $\sigma_1,\ldots,\sigma_k$ such that $\Cmax(\sigma_a)\le (2k+1+\varepsilon_0)T$ for every $a\in[k]$, and $d(\sigma_a,\sigma_b)\ge (1-\varepsilon_0)Q$ for every $1\le a<b \le k$.
\end{restatable}
\begin{proof}
  By \Cref{lem:correct-guess-bounded}, the residual $k$-tuple LP is feasible and all remaining diversity coefficients are at most~$Q/h$ for active pairs.
    
    At each iteration, we compute an extreme-point solution of the current residual LP. 
    If a variable has value~$0$, we remove it. 
    If a variable has value~$1$, we fix the corresponding tuple, update all residual right-hand sides, set the other variables of that job to~$0$, and remove the job. 
    If an active diversity constraint has non-positive residual right-hand side, we remove it. 
    If all remaining variables are fractional, we apply \Cref{lem:small-mass-constraint} and drop a remaining load constraint or active diversity constraint of scaled mass at most~$2\Lambda$. 
    Since job equalities are never dropped, the final solution selects exactly one tuple per unguessed job. 
    Together with the guessed tuples, this gives complete schedules.
    
    Consider a load constraint~$(a,i)$. 
    If it is never dropped, the final integral solution satisfies its residual target. 
    If it is dropped, consider the moment when it is dropped. 
    The current fractional solution satisfies the residual load bound, and the total remaining load mass in this constraint is at most~$2\Lambda$. 
    Hence, after the constraint is dropped, the additional load that can still be assigned to~$(a,i)$ is at most~$2\Lambda$. 
    Adding the guessed load, we get that the final load is at most
    \begin{align*}
        T^{\mathrm{fix}}_{a,i} + T^{\mathrm{res}}_{a,i}+ 2\Lambda
        \;=\; T + 2(1+\gamma)kT
        \;=\; (2k+1+\varepsilon_0) T  \enspace .
    \end{align*}
    
    Now consider a pair~$(a,b)$. 
    If it is inactive, then~$Q^{\mathrm{fix}}_{a,b}\ge Q$, so the final diversity is already at least~$Q$. 
    If it is active and its constraint is never dropped, or if it is removed because its residual target is non-positive, then the residual target is met. 
    If it is dropped with positive residual target, the remaining scaled mass is at most~$2\Lambda$. 
    Hence the unscaled diversity that we can lose after the constraint is dropped is at most
    \begin{align*}
        \frac{2\Lambda}{\beta} 
        \;=\; \frac{2 (1+\gamma) kT}{\gamma kT h / (\binom{k}{2}Q)}
        \;=\; \frac{2 (1+\gamma) \binom{k}{2} Q}{\gamma h} 
        \;\le\; \varepsilon_0 Q \enspace .
    \end{align*}
    Thus every active pair has final diversity at least $Q^{\mathrm{fix}}_{a,b} + Q^{\mathrm{res}}_{a,b} - \varepsilon_0 Q \ge (1-\varepsilon_0)Q$.
    
    The run time of the rounding step is polynomial for fixed~$k$ and~$\varepsilon_0$, because the LP has~$\mathcal{O}(nm^k)$ variables and each iteration fixes a variable, removes a job, or drops a constraint.
\end{proof}

\paragraph*{Guessing the diversity target and proof of \Cref{thm:max-min-ktuple}.}

\begin{proof}[Proof of \Cref{thm:max-min-ktuple}]
    If $T=0$, we replace every positive processing time by~$2$
    and use makespan target~$1$. This preserves the eligible pairs and the feasible schedules, so the diversity optimum is unchanged.
    The rounding procedure assigns jobs only to eligible machines; hence every output schedule has makespan~$0$ in the original
    instance. We may therefore assume $T>0$.
    
    Let~$Q^*=\OPT_k^{\min}(T)$ and set~$\varepsilon_0=\varepsilon/3$. 
    We always try the value $Q=0$ and retain the resulting $k$ schedules as an initial candidate family, called the baseline family.
    For this run, we ignore the diversity constraints and apply the same iterative rounding with only load constraints.
    Since each tuple has total load mass at most~$kT$, the same counting argument gives schedules of makespan at most~$(2k+1)T$.
    
    Let $W=\sum_j w_j$.
    If $W=0$, return the baseline family.
    Otherwise, let \mbox{$w_{\min}=\min\{w_j \mid w_j>0\}$}.
    Enumerate $Q=W/(1+\varepsilon_0)^\ell$ for $\ell=0,1,\ldots$, including and stopping at the first value with $Q\leq w_{\min}$.
    The baseline covers $Q^*=0$. If $Q^*>0$, then $w_{\min}\leq Q^*\leq W$, so some value~$\widehat Q$ satisfies
    \begin{align}\label{eq:geometric}
        \widehat{Q} \;\le\; Q^* \;<\; (1+\varepsilon_0)\widehat{Q} \enspace .
    \end{align}
    
    For a fixed guessed value~$Q>0$, we enumerate for every pair~$(a,b)$ up to~$h$ jobs together with their full $k$-tuples. We discard guesses that are inconsistent or whose fixed jobs already violate a makespan constraint.
    For each remaining candidate, we form the residual LP, discard the candidate if this LP is infeasible, and otherwise apply the same iterative-rounding procedure. 
    The number of candidates is~$(nm^k)^{\mathcal{O}(\binom{k}{2}h)}$.
    
    For the guessed value~$Q = \widehat{Q}$, consider the candidate that agrees with a set of optimal schedules with max-min diversity~$Q^*$. 
    Since~$\widehat{Q} \le Q^*$, these schedules have pairwise diversity at least~$\widehat{Q}$, and \Cref{lem:rounding-after-guessing} applies. 
    Thus we obtain schedules of makespan at most~$(2k+1+\varepsilon_0) T$ and pairwise diversity at least $(1-\varepsilon_0) \widehat{Q}$.
    
    Using~$Q=\widehat{Q}$ and \eqref{eq:geometric}, we have that
    \begin{align*}
        d_{\min}(\sigma_1,\ldots,\sigma_k)
        \;\ge\; (1-\varepsilon_0)\widehat{Q}
        \;\ge\; \frac{1-\varepsilon_0}{1+\varepsilon_0} Q^* 
        \;\ge\; (1-\varepsilon) Q^*  .
    \end{align*}
    Moreover, the makespan is at most~$(2k+1+\varepsilon_0)T \le (2k+1+\varepsilon) T$.
    
    Since~$h=\mathcal{O}\left(\frac{k^3}{\varepsilon_0^2}\right) =\mathcal{O}\left(\frac{k^3}{\varepsilon^2}\right)$, the enumeration takes time 
    \begin{align*}
       \bigl(nm^k\bigr)^{\mathcal{O}\bigl(\binom{k}{2} h \bigr)}
        \;=\; \bigl(nm^k\bigr)^{\mathcal{O}\bigl(k^5/\varepsilon^2\bigr)} \enspace .
    \end{align*}
    Each rounding run is polynomial in the size of the $k$-tuple LP, and the geometric search contributes only a polynomial factor in the input encoding length~$|I|$. 
    Therefore the total run time is~$(nm^k)^{\mathcal{O}(k^5/\varepsilon^2)}\cdot |I|^{O(1)}$.
    Returning a solution of maximum actual max-min diversity over all guesses and the baseline family yields the theorem.
\end{proof}

\subsection{Unrelated-machine scheduling with covering constraints}
\label{app:covering-constraints}

\paragraph*{Guessing the heavy jobs.}
\label{subsubsec:guess-heavy-contributions}
Let~$\tau$ denote an (unknown) integral feasible solution of~\eqref{eq:max-min-div-LP}, which exists by assumption, and fix~$0<\varepsilon<1$.
We say a job~$j$ is \emph{heavy} for constraint~$a$ if it is among the~$h$ jobs of largest positive contribution~$q^a_{\tau(j),\,j}$ under~$\tau$, where
\begin{align*}
    h \;=\; \left\lceil \frac{2(1+\varepsilon)\, r}{\varepsilon^2} \right\rceil  .
\end{align*}

The assignment~$\tau$ is unknown, hence the algorithm cannot identify the heavy jobs directly.
Instead, we enumerate all possible choices.
For each constraint~$a\in[r]$ we guess up to~$h$ jobs together with their assigned machines.
We call such a guessed partial assignment a \emph{candidate}.
We immediately discard candidates that assign one job to multiple different machines or whose guessed jobs already violate a makespan target.
For each constraint, the number of possible guesses is at most
$\sum_{\ell=0}^{\min\{h,n\}} \binom n\ell m^\ell \le (n+1)(nm)^h$.
Thus the total number of candidates is at most
\begin{align*}
  \bigl((n+1)(nm)^h\bigr)^r =
  (nm)^{\mathcal{O}(rh)} =
  (nm)^{\mathcal{O}(r^2/\varepsilon^2)}  .
\end{align*}
Among these candidates is the one that, for every constraint~$a\in[r]$, assigns the heavy contributors of~$a$ to their machines under~$\tau$, or all positively contributing jobs if there are fewer than~$h$ of them.
We fix this candidate for the remainder of the analysis.
This candidate does not need to reconstruct~$\tau$ on the other jobs. It only fixes the heaviest contributors of~$\tau$, which is precisely the property that bounds the rounding loss.

For this fixed candidate, let~$G$ be the set of guessed jobs, where~$j$ is placed on guessed machine~$\tau(j)$.
We initialize the \emph{residual makespan targets} and \emph{residual covering targets} as
\begin{align*}
    T_i^{\mathrm{res}} \;=\; T_i - \sum_{j\in G:\,\tau(j)=i} p_{ij}
    \qquad\text{and}\qquad
    Q_a^{\mathrm{res}} \;=\; \Bigl(Q_a - \sum_{j\in G} q^a_{\tau(j),\,j}\Bigr)^+  .
\end{align*} 
By the discard rule, $T_i^{\mathrm{res}}\ge 0$ for every machine~$i$. The residual makespan and covering targets are updated as the iterative rounding assigns more jobs.
A constraint with~$Q_a^{\mathrm{res}} = 0$ is already covered by~$G$ and is discarded from the residual LP.
Furthermore, let~$\mathcal{A} = \{a\in [r]\mid Q_a^\mathrm{res}>0\}$ be the remaining \emph{active} covering constraints.
The \emph{residual LP} has a variable~$x_{ij}$ for each unguessed job~$j$ and machine~$i\in\Gamma(j)$, assignment constraints~$\sum_{i\in\Gamma(j)} x_{ij} = 1$, capacities~$\sum_{j} p_{ij} x_{ij} \leq T_i^{\mathrm{res}}$, and active targets~$\sum_{j}\sum_{i} q^a_{ij} x_{ij} \geq Q_a^{\mathrm{res}}$ for~$a\in\mathcal{A}$.
Restricting~$\tau$ to the unguessed jobs gives a feasible integral point, so the residual LP is feasible.

We now show that in the residual LP we can safely remove every variable~$x_{ij}$ whose contribution~$q^a_{ij}$ is large relative to the target~$Q_a$ of some active constraint~$a\in\mathcal{A}$, namely $q^a_{ij} > Q_a/h$.
We show that after removing these high-contribution variables, the restriction of~$\tau$ to the unguessed jobs remains feasible for the residual LP.
\begin{lemma}[Removing high-contribution pairs is safe]
    \label{lem:high-contribution-pair-removal}
    Delete from the residual LP every variable~$x_{ij}$ with~$q^a_{ij} > Q_a/h$ for some active~$a\in\mathcal{A}$.
    Then no assignment pair used by~$\tau$ is deleted.
    Consequently the residual LP remains feasible, witnessed by~$\tau$ restricted to the unguessed jobs, and every surviving variable satisfies~$q^a_{ij}\leq Q_a/h$ for all~$a\in\mathcal{A}$.
\end{lemma}
\begin{proof}
    Suppose that there is an active covering constraint~$a\in \mathcal{A}$ and an unguessed job~$j\not\in G$ such that~$\tau(j) = i^\star$ and~$q^a_{i^\star j} > Q_a/h$.
    We show that this is impossible by deriving a contradiction.
    If~$\tau$ uses fewer than~$h$ jobs with positive contribution to covering constraint~$a$, then the candidate guesses all of them, including~$j$, a contradiction. Hence,~$\tau$ uses at least~$h$ jobs with positive contribution to covering constraint~$a$. The candidate guesses the~$h$ largest positive contributors for constraint~$a$ under~$\tau$. Since the unguessed assignment pair~$(i^\star,j)$ has contribution~$q^a_{i^\star j}>Q_a/h$, each guessed contributor has contribution at least~$q^a_{i^\star j}$. Together they therefore contribute more than~$h (Q_a/h)=Q_a$ to constraint~$a$.
    It follows that~$\sum_{j'\in G} q^a_{\tau(j'),\,j'} > Q_a$, hence~$Q_a^{\mathrm{res}} = 0$, contradicting that~$a$ is active.
    Therefore no pair of~$\tau$ is deleted, so~$\tau$ restricted to the unguessed jobs remains feasible, and the surviving variables are exactly those with~$q^a_{ij}\leq Q_a/h$ for every active~$a\in\mathcal{A}$.
\end{proof}

For an arbitrary candidate, if the resulting residual LP is infeasible, we discard the candidate. Otherwise, we run the iterative-rounding procedure described below.

\paragraph*{Iterative rounding.}
\label{subsubsec:iterative-rounding}

After removing the variables corresponding to high-contribution pairs, the algorithm rounds the residual LP by iterative relaxation.
We maintain a fractional extreme point~$x$ of the current residual LP and repeatedly apply one of the relaxation rules below.
We recompute an extreme point of the updated LP after each relaxation.
Throughout, we write~$E$ for the set of pairs~$(i,j)$ such that the variable~$x_{ij}$ is part of the residual LP.
Moreover, for a machine~$i$, we write~$d(i)$ for the number of unassigned jobs that still have~$i$ as an eligible machine, i.e., the number of remaining variables~$x_{ij}$.
A job is \emph{assigned} once its variable is fixed to~$1$, at which point the residual capacity~$T_i^{\mathrm{res}}$ is decreased by the load of this job, and each residual covering target~$Q_a^{\mathrm{res}}$ becomes~$(Q_a^{\mathrm{res}}-q^a_{ij})^+$.

The rules are as follows.
If a variable~$x_{ij}$ is equal to~$0$, it is removed from the LP.
If a variable~$x_{ij}$ is equal to~$1$, then job~$j$ is assigned to machine~$i$, and we update~$T_i^{\mathrm{res}}$ and the~$Q_a^{\mathrm{res}}$ accordingly and remove the variables corresponding to~$j$.
If an active constraint~$a$ is covered by the assigned jobs, that is,~$Q_a^{\mathrm{res}} = 0$, then we remove it from~$\mathcal{A}$ and the LP.
Finally, if some machine~$i$ has~$d(i) \leq 1 + \varepsilon + \sum_{j} x_{ij}$, then we remove the makespan constraint of machine~$i$ from the LP.
This last rule is inspired by the iterative rounding procedure of Korupolu, Meyerson, Rajaraman, and Tagiku~\cite{KorupoluMRT2015}. 

Each rule strictly decreases the number of variables or constraints, so the process terminates.
Next we show that once none of these rules applies, the remaining number of variables is small, and since all jobs have low contribution to the covering constraints, completing the remaining assignment can lose at most an~$\varepsilon$ fraction of each covering target.

\begin{restatable}{lemma}{supportboundlemma}
    \label{lem:support-bound}
    Let~$x$ be an extreme point of the residual LP at which none of the rules above apply.
    Then $|E| \;\leq\;  \frac{2(1+\varepsilon)}{\varepsilon}\cdot r$.
\end{restatable}
\begin{proof}
    If $E=\varnothing$, the claim is immediate. We therefore assume $E\ne\varnothing$. Let~$M'$ be the set of machines still having a makespan constraint and let~$J'$ be the set of unassigned jobs.
    We denote by~$M'' = M\setminus M'$ the machines whose makespan constraint has been dropped.
    Every remaining variable satisfies~$0<x_{ij}<1$, because the first two rules do not apply.
    Moreover, since the last rule does not apply, we have~$d(i) > 1+\varepsilon+\sum_{j} x_{ij}$ for every machine~$i\in M'$, where the sum is over the unassigned jobs with a variable on~$i$.

    \pagebreak

    Each remaining variable is counted by exactly one machine, so~$\sum_{i\in M} d(i) = |E|$.
    Using the degree bound for~$M'$ and~$\sum_{i\in M} x_{ij} = 1$ yields
    \begin{align}
        |E|
        \;&=\; \sum_{i\in M'} d(i) + \sum_{i \in M''} d(i) 
        \;>\; (1+\varepsilon) |M'| + \sum_{i\in M'}\sum_{j\in J' : \, (i,j) \in E} x_{ij} + \sum_{i\in M''} \sum_{j\in J' :\, (i,j)\in E} x_{ij} \notag \\
        \;&=\; (1+\varepsilon) |M'| + \sum_{j\in J'} \sum_{i\in M :\, (i,j) \in E} x_{ij} 
        \;=\; (1+\varepsilon) |M'| + |J'|  \enspace . \label{eq:support-edge-lb}
    \end{align}

    All remaining variables are strictly fractional, so no variable-bound constraint is tight.
    Thus, the extreme point is determined by linearly independent tight constraints among the assignment constraints, the remaining makespan constraints, and the active covering constraints.
    There are at most~$|J'|$ assignment constraints, at most~$|M'|$ remaining makespan constraints, and at most~$r$ active covering constraints, hence
    \begin{align}
        |E| \;\leq\; |J'| + |M'| + r  \enspace .\label{eq:support-rank}
    \end{align}
    Moreover, each unassigned job is split over at least two machines, because it has~$0<x_{ij}<1$ and~$\sum_i x_{ij}=1$, so it accounts for at least two variables, hence~$|E|\geq 2|J'|$.
    If~$M'=\emptyset$, then \eqref{eq:support-rank} and~$|J'|\leq |E|/2$ imply~$|E|\leq |E|/2+r$, and hence~$|E|\leq 2r\leq \frac{2(1+\varepsilon)}{\varepsilon}r$. 
    Thus, we may assume~$M'\ne\emptyset$.
    We combine~\eqref{eq:support-edge-lb}, \eqref{eq:support-rank} and~$|J'|\leq \tfrac12 |E|$ to obtain
    \begin{align*}
        |E| \;\leq\; |J'| + |M'| + r
        &\;<\; |J'| + \frac{|E|-|J'|}{1+\varepsilon} + r
        \;=\; \frac{\varepsilon}{1+\varepsilon} \cdot |J'| + \frac{1}{1+\varepsilon} \cdot |E| + r \\
        &\;\leq\; \frac{\varepsilon}{2(1+\varepsilon)} \cdot |E| + \frac{1}{1+\varepsilon}\cdot |E| + r  \enspace .
    \end{align*}
    Rearranging yields~$\tfrac{\varepsilon}{2(1+\varepsilon)}\cdot |E| \leq r$, that is,~$|E|\leq \tfrac{2(1+\varepsilon)}{\varepsilon}\cdot r$.
\end{proof}

Thus, when none of the rounding rules apply, the process ends with at most~$|E| \leq \tfrac{2(1+\varepsilon)}{\varepsilon}r$ variables still fractional.
Crucially, the covering constraints no longer matter, because each surviving variable contributes at most~$Q_a/h$ to any active constraint~$a\in \mathcal{A}$, and hence by the choice of~$h$, the total fractional contribution is at most~$\varepsilon Q_a$.
This means that any assignment of the remaining fractional variables loses at most a fraction~$\varepsilon$ of each covering target.
We discard the remaining covering constraints.
The residual assignment LP now only has assignment and makespan constraints, and the surviving fractional solution is a feasible point.
Rounding this LP the same way as Shmoys and Tardos~\cite{ShmoysT1993} did assigns each remaining job to an eligible machine and increases the load on every machine whose makespan constraint remains by at most the largest processing time of a remaining job eligible on that machine, which is at most its target. 
We show in~\Cref{lem:makespan-guarantee} that the resulting load stays within~$(2+\varepsilon)T_i$ on each machine.
Before turning to the makespan, we first prove that the covering guarantee holds for any assignment of the remaining variables.

\begin{restatable}[Covering guarantee]{lemma}{coveringguaranteelemma}
\label{lem:covering-guarantee}
  When the iterative rounding terminates, dropping all remaining active covering constraints and completing the assignment yields for every~$a\in [r]$ that
    \begin{align*}
        \sum_{j\in J}\sum_{i\in\Gamma(j)} q^a_{ij} X_{ij} \;\geq\; (1-\varepsilon) Q_a  \enspace .
    \end{align*}
\end{restatable}
\begin{proof}
    Fix a constraint~$a\in[r]$.
    If~$a$ was discarded during the rounding or inactive from the start, then~$Q_a^{\mathrm{res}}=0$, meaning the guessed and assigned jobs already contribute the full target~$Q_a$, hence the bound holds.
    So suppose~$a$ is still active when the rounding terminates.

    By \Cref{lem:high-contribution-pair-removal} every surviving variable satisfies~$q^a_{ij}\leq Q_a/h$.
    Hence, using the bound of \Cref{lem:support-bound}, the contribution to~$a$ carried by the remaining fractional variables is
    \begin{align*}
        \sum_{(i,j)\in E} q^a_{ij}x_{ij}
        \;\leq\; \frac{Q_a}{h}\cdot |E|
        \;\leq\; \frac{Q_a}{h}\cdot\frac{2(1+\varepsilon)}{\varepsilon}\cdot r
        \;\leq\; Q_a\cdot\frac{\varepsilon^2}{2(1+\varepsilon)\,r}\cdot\frac{2(1+\varepsilon)\,r}{\varepsilon}
        \;=\; \varepsilon\cdot Q_a \enspace .
    \end{align*}
    Therefore the fractional variables carry at most~$\varepsilon Q_a$ of constraint~$a$, and dropping them decreases the covering value by at most this much.
    The guessed and assigned jobs supply the remaining target~$Q_a - Q_a^{\mathrm{res}}$, and the fractional solution is at least~$Q_a^{\mathrm{res}}$ before rounding, so after discarding the fractional contribution the integral assignment collects at least
    \begin{align*}
        (Q_a - Q_a^{\mathrm{res}}) + (Q_a^{\mathrm{res}} - \varepsilon Q_a)
        \;=\; (1-\varepsilon) Q_a  \enspace .
    \end{align*}
    Completing the assignment of the remaining fractional jobs only adds contribution, so the bound is preserved.
\end{proof}

It remains to bound the makespan. 
The residual capacity~$T_i^{\mathrm{res}}$ tracks the budget left on machine~$i$.
Initially~$T_i^{\mathrm{res}}$ is equal to~$T_i$ minus the guessed load, and it decreases by~$p_{ij}$ whenever a job is assigned to machine~$i$. 
Thus, as long as the makespan constraint of machine~$i$ remains present, the assigned load charged to~$i$ never exceeds~$T_i$, and the increase over the fractional load present when this constraint is dropped is bounded in~\Cref{lem:makespan-guarantee} by at most~$(1+\varepsilon)T_i$.
\begin{restatable}[Makespan guarantee]{lemma}{makespanguaranteelemma}
\label{lem:makespan-guarantee}
    The final assignment~$X\in \{0,1\}^{m\times n}$ satisfies
    \[\sum_{j:\,i\in\Gamma(j)} p_{ij} X_{ij} \leq (2+\varepsilon) T_i\]
    for every machine~$i \in M$.
\end{restatable}
\begin{proof}
    Fix a machine~$i$.
    The makespan constraint of machine~$i$ is either dropped before all rounding rules have been applied or not.
    First suppose that machine~$i$ is dropped, and let~$x$ be the fractional solution at the moment it is dropped from the residual LP.
    Then there are at most~$d(i) \leq 1+\varepsilon + \sum_{j} x_{ij}$ jobs with a variable to machine~$i$.
    The jobs already assigned to machine~$i$ have load equal to~$T_i-T_i^\mathrm{res}$.
    Furthermore,~$\sum_{j} p_{ij} x_{ij} \leq T_i^{\mathrm{res}}$.
    All variables considered here are eligible on machine~$i$, so $p_{ij}\le T_i$.
    Hence, the increase in load on machine~$i$ after rounding these jobs is bounded by
    \begin{align*}
        \sum_{j} p_{ij} (1-x_{ij})
        \;\leq\; T_i \Bigl( d(i) - \sum_{j} x_{ij}\Bigr)
        \;\leq\; (1+\varepsilon) T_i  \enspace .
    \end{align*}
    Thus, the total load on machine~$i$ for the final integral assignment is at most the load of the jobs already assigned to~$i$ at the moment that the makespan constraint is dropped, plus the load of the fractional jobs at that moment, plus the potential increase of the fractional jobs when they are rounded up, that is,
    \begin{align*}
        (T_i - T_i^{\mathrm{res}}) + T_i^{\mathrm{res}} + (1+\varepsilon) T_i \;=\; (2+\varepsilon) T_i  \enspace .
    \end{align*}

    In the second case, the makespan constraint of machine~$i$ is never dropped.
    Let~$x$ be the final fractional solution at the moment that none of the rounding rules apply.
    Then, as mentioned in~\Cref{lem:covering-guarantee}, we can drop all covering constraints and are left with a generalized assignment problem.
    We can round the job assignments using the rounding technique of Shmoys and Tardos~\cite{ShmoysT1993}, violating the remaining makespan constraint on machine~$i$ by at most the largest processing time of a remaining job eligible on~$i$.
    This largest processing time is at most~$T_i$.
    Hence, the total load on machine~$i$ is at most
    $(T_i-T_i^{\mathrm{res}})+T_i^{\mathrm{res}}+T_i=2T_i$.
\end{proof}

\paragraph*{Proof of \texorpdfstring{\Cref{prop:rounding-covering-constraints}}{Proposition \ref{prop:rounding-covering-constraints}}.}
\label{subsubsec:proof-of-rounding-theorem}

We now assemble the preceding lemmas into a proof of~\Cref{prop:rounding-covering-constraints}.

\roundingcoveringconstraints*

\begin{proof}[Proof of \Cref{prop:rounding-covering-constraints}]
    If~$r=0$, there are no covering constraints, and the claim directly follows from the generalized-assignment rounding of Shmoys and Tardos~\cite{ShmoysT1993}, which already yields makespan at most~$2T_i$. We therefore assume~$r\geq1$ and run the algorithm described above.

    The algorithm enumerates all candidates and, for each one that is not discarded, performs the iterative rounding and outputs an integral assignment.
    It returns an assignment satisfying the makespan and covering bounds.
    To prove correctness it therefore suffices to exhibit one candidate whose run produces such an assignment.
    Let~$G^\star$ be the candidate that, for each constraint~$a\in[r]$, selects the heavy contributors of~$a$ together with their machines under~$\tau$.
    Since~$G^\star$ agrees with the feasible assignment~$\tau$, it is consistent and its fixed load on every machine is at most~$T_i$, so it is not discarded.

    Consider the run on~$G^\star$.
    By \Cref{lem:high-contribution-pair-removal}, deleting the high-contribution variables keeps the residual LP feasible, witnessed by~$\tau$ restricted to the unguessed jobs, and leaves every surviving contribution below~$Q_a/h$.
    The iterative rounding terminates, and by \Cref{lem:support-bound} we have~$|E|\leq \tfrac{2(1+\varepsilon)}{\varepsilon}r$.
    Discarding the remaining covering constraints and completing the assignment, \Cref{lem:covering-guarantee} gives~$\sum_{j}\sum_{i} q^a_{ij}X_{ij}\geq(1-\varepsilon)Q_a$ for every~$a\in[r]$, and \Cref{lem:makespan-guarantee} gives~$\sum_{j} p_{ij}X_{ij}\leq(2+\varepsilon)T_i$ for every machine~$i$.
    Since each job is assigned to exactly one eligible machine, we also have~$\sum_{i\in\Gamma(j)}X_{ij}=1$ for every job~$j$.
    Hence the run on~$G^\star$ produces an assignment meeting the makespan and covering bounds.

    It remains to bound the run time.
    There are at most~$(nm)^{\mathcal{O}(rh)} = (nm)^{\mathcal{O}(r^2/\varepsilon^2)}$ candidates.
    For each candidate the iterative rounding performs polynomially many steps, each recomputing an extreme point by solving a linear program of polynomial size, and assigning the remaining variables also takes polynomial time. The total run time is therefore~$(nm)^{\mathcal{O}(r^2/\varepsilon^2)}\cdot|I|^{O(1)}$.
\end{proof}

\subsection{Low-makespan max-min diversity}
\label{app:low-makespan-max-min}

We are now ready to give the full proof of \Cref{thm:max-min-2-makespan}.
\maxmintwomakespan*

\begin{proof}
    Set $\varepsilon_0=\varepsilon/3$. We first prove a threshold statement. Fix a value~$Q\le \OPT_k^{\min}(T)$. We show that, given such a threshold~$Q$, one can compute~$k$ schedules of makespan at most~$(2+\varepsilon_0)T$ and max-min diversity at least~$\frac12(1-\varepsilon_0)Q$.

    The algorithm constructs the schedules one by one. After schedules~$\sigma_1,\dots,\sigma_r$ have been constructed, we maintain that each of them has makespan at most~$(2+\varepsilon_0)T$ and that every pair among them has diversity at least~$\frac12(1-\varepsilon_0)Q$.

    For the first schedule, we apply~\Cref{prop:rounding-covering-constraints} with no covering constraints and makespan targets~$T_i=T$ for all machines~$i$. Since~$T$ is feasible, the corresponding LP has an integral feasible solution. Thus, we obtain a schedule~$\sigma_1$ with~$\Cmax(\sigma_1)\le (2+\varepsilon_0)T$.

    Now suppose that~$1<r\le k$ and that schedules~$\sigma_1,\dots,\sigma_{r-1}$ have already been constructed. 
    Let~$\pi_1,\dots,\pi_k$ be an optimal solution for~$\OPT_k^{\min}(T)$.
    Thus, $\Cmax(\pi_a)\le T$ for every~$a\in[k]$, and $d(\pi_a,\pi_{a'})\ge \OPT_k^{\min}(T)\ge Q$ for every distinct $a,a'\in[k]$.
    We claim that there exists~$a\in[k]$ such that
    \begin{align*}
        d(\pi_a,\sigma_b)\ge \frac12 Q
        \qquad \forall b\in[r-1]  \enspace .
    \end{align*}
    Suppose not.
    Then for every~$a\in[k]$, there exists some~$b(a)\in[r-1]$ such that
    \begin{align*}
        d(\pi_a,\sigma_{b(a)})<\frac12 Q  \enspace .
    \end{align*}
    Since~$r-1<k$, the pigeonhole principle implies that two distinct optimal schedules~$\pi_a$ and~$\pi_{a'}$ are mapped to the same constructed schedule~$\sigma_b$.
    By the triangle inequality,
    \begin{align*}
        d(\pi_a,\pi_{a'}) \;\le\; d(\pi_a,\sigma_b)+d(\sigma_b,\pi_{a'}) <\; \frac12 Q+\frac12 Q \;=\; Q \enspace .
    \end{align*}
    This contradicts~$d(\pi_a,\pi_{a'})\ge \OPT_k^{\min}(T)\ge Q$.
    Hence, such a schedule~$\pi_a$ exists.

    Fix such a schedule~$\pi_a$.
    The schedule~$\pi_a$ has makespan at most~$T$ and diversity at least~$\frac12 Q$ from each of~$\sigma_1,\dots,\sigma_{r-1}$.
    For each previously constructed schedule~$\sigma_b$, define covering coefficients
    \begin{align*}
        q^b_{ij}=w_j\mathbf 1[i\ne \sigma_b(j)]  \enspace .
    \end{align*}
    Then~$\pi_a$ is an integral feasible solution of~\eqref{eq:max-min-div-LP} with makespan targets~$T_i=T$ for all machines~$i$ and covering targets~$\frac12 Q$, one covering constraint for each previously constructed schedule.
    Applying~\Cref{prop:rounding-covering-constraints} with parameter~$\varepsilon_0$ gives a schedule~$\sigma_r$ with~$\Cmax(\sigma_r)\le (2+\varepsilon_0)T$ and, for every~$b\in[r-1]$,
    \begin{align*}
        d(\sigma_b,\sigma_r) \;\ge\; (1-\varepsilon_0)\cdot \frac12 Q \;=\; \frac12(1-\varepsilon_0)Q \enspace .
    \end{align*}
    Thus, after adding~$\sigma_r$, all constructed schedules still satisfy the same makespan and pairwise-diversity guarantees.
    By induction, after~$k$ steps we obtain~$k$ schedules of makespan at most \mbox{$(2+\varepsilon_0)T$} and max-min diversity at least~$\frac12(1-\varepsilon_0)Q$.
    This proves the threshold statement.

    It remains to remove the assumption that a suitable threshold~$Q$ is known.
    Let $Q^*=\OPT_k^{\min}(T)$.
    We always run the threshold algorithm with $Q=0$ and retain its output as a baseline family. Let $W=\sum_{j\in J}w_j$.
    If $W=0$, return this family.
    Otherwise, let $w_{\min}=\min\{w_j \mid w_j>0\}$.
    Enumerate $Q=W/(1+\varepsilon_0)^\ell$ for $\ell=0,1,\ldots$, including and stopping at the first value with $Q\leq w_{\min}$.
    The baseline covers $Q^*=0$.
    If $Q^*>0$, then $w_{\min}\leq Q^*\leq W$, so one guessed threshold~$\widehat Q$ satisfies
    \begin{align*}
        \widehat Q\le Q^* < (1+\varepsilon_0)\widehat Q  \enspace .
    \end{align*}
    We run the threshold algorithm for every guessed threshold.
    Runs that do not return~$k$ schedules satisfying the required makespan and pairwise-diversity bounds for the guessed threshold are discarded.
    Among the remaining outputs and the baseline family, we return a family with maximum actual max-min diversity.

    For the guessed threshold~$\widehat Q$, the threshold statement applies. Therefore, one of the runs returns schedules of makespan at most~$(2+\varepsilon_0)T\le (2+\varepsilon)T$ and max-min diversity at least
    \begin{align*}
        \frac12(1-\varepsilon_0)\widehat Q \;\ge\; \frac12\cdot \frac{1-\varepsilon_0}{1+\varepsilon_0} Q^* \ge\; \frac12(1-\varepsilon)Q^*\;=\; \frac12(1-\varepsilon)\OPT_k^{\min}(T) \enspace .
    \end{align*}
    The second inequality follows from~$\varepsilon_0=\varepsilon/3$ and~$\varepsilon\in(0,1)$.

    Finally, we bound the run time. For a fixed threshold, the sequential construction invokes~\Cref{prop:rounding-covering-constraints} exactly~$k$ times. Each invocation has at most~$k-1$ covering constraints and uses parameter~$\varepsilon_0$.
    Therefore, each threshold run takes time $(nm)^{\mathcal{O}(k^2/\varepsilon^2)}\cdot |I|^{O(1)}$.
    The number of guessed thresholds is $\mathcal O\left(1+\frac{1}{\varepsilon}\log\frac{W}{w_{\min}}\right)$, which is absorbed in the stated bound. Hence the total run time is $(nm)^{\mathcal{O}(k^2/\varepsilon^2)}\cdot |I|^{O(1)}$. 
\end{proof}

\subsection{Hardness for max-min diversity}
\label{app:max-min-hardness}
We prove the hardness result using two machines and all processing times equal to~$0$.
Thus, every schedule has makespan~$0$, and only the diversity requirement matters.

\maxminhardness*
\begin{proof}
  We reduce from \textsc{Partition}. 
  Let $a_1, \dots, a_n$ be an instance of \textsc{Partition} with \mbox{$\sum_{j=1}^n a_j = 2B$}. 
  Without loss of generality, we assume that $0 < a_j \le B$ for every $j$.

  For each $a_j$, we create a job~$j$ of weight~$w_j=a_j$. We also create one additional dummy job~$d$ of weight~$w_d=B$. Thus, the total weight of all jobs is~$3B$.

  We use exactly two machines, $M=\{1,2\}$. Every job, including the dummy job, may be assigned to either machine, and we set $p_{ij}=0$ for every machine~$i$ and job~$j$. Consequently, every assignment is a feasible schedule with makespan~$0$.
    
  We set the diversity value $Q$ equal to $2B$. 
  We show that the \textsc{Partition} instance is a ``yes''-instance if and only if there exist three feasible schedules $\sigma_1, \sigma_2, \sigma_3$ such that
  \begin{align*}
    d(\sigma_a,\sigma_b)\ge 2B,
      \qquad \forall 1\le a<b\le 3  \enspace .
  \end{align*}
    
  First suppose that the \textsc{Partition} instance is a ``yes''-instance. 
  Then there exists a set $S \subseteq \{1, \dots, n\}$ such that $\sum_{j\in S}a_j=B$. 
  Let $\overline S=\{1,\ldots,n\}\setminus S$. 
  Since $\sum_{j=1}^n a_j = 2B$, we also have that $\sum_{j\in \overline S}a_j=B$.
  We define three schedules $\sigma_1, \sigma_2, \sigma_3$ as follows.
  For the dummy job, set
  \begin{align*}
        \sigma_1(d)=1,
        \qquad
        \sigma_2(d)=1,
        \qquad
        \sigma_3(d)=2 \enspace .
  \end{align*}
    
  For each job $j \in S$ set 
  \begin{align*}
    \sigma_1(j)=1,
        \qquad
        \sigma_2(j)=2,
        \qquad
        \sigma_3(j)=2 \enspace .
  \end{align*}
    
  For each job $j \in \overline S$ set
  \begin{align*}
    \sigma_1(j)=2,
        \qquad
        \sigma_2(j)=1,
        \qquad
        \sigma_3(j)=2 \enspace .
  \end{align*}
  The three groups $S$, $\overline S$, and $\{d\}$ each have total weight $B$, and each pair of schedules differs on exactly two of them. 
  Hence,
  \begin{align*}
    d(\sigma_1,\sigma_2)= 2B, \qquad
    d(\sigma_1,\sigma_3)= 2B, \qquad
    d(\sigma_2,\sigma_3)= 2B \enspace .
  \end{align*}
  Therefore, the constructed instance admits three feasible schedules with minimum pairwise diversity at least $2B$.
    
  Conversely, suppose there exist three schedules $\sigma_1, \sigma_2, \sigma_3$ such that 
  \begin{align*}
    d(\sigma_a,\sigma_b)\ge 2B,
        \qquad \forall 1\le a<b\le 3 \enspace .
  \end{align*}
    
  Since there are only two machines, for each job either all three schedules assign it to the same machine, or exactly one schedule assigns it to a different machine from the other two.
    
  Let~$G_0$ be the set of jobs on which all three schedules agree. 
  For $i \in \{1,2,3\}$, let~$G_i$ be the set of jobs for which~$\sigma_i$ is the unique schedule that differs from the other two. 
  Define 
  \begin{align*}
        W_i = \sum_{j \in G_i} w_j, \qquad i \in \{0,1,2,3\} \enspace .
  \end{align*}
  Then
  \begin{align*}
    W_0+W_1+W_2+W_3 = 3B \enspace .
  \end{align*}
  Jobs in~$G_0$ do not contribute any positive pairwise diversity.
  A job in $G_1$ contributes to the two pairwise diversities involving $\sigma_1$, and
  analogously for $G_2$ and $G_3$.
  Therefore,
  \begin{align*}
        d(\sigma_1,\sigma_2) = W_1+W_2, \qquad
        d(\sigma_1,\sigma_3) = W_1+W_3, \qquad
        d(\sigma_2,\sigma_3)= W_2+W_3 \enspace .
  \end{align*}
  Since all three pairwise diversities are at least~$2B$, summing the three inequalities gives
  \begin{align*}
    2(W_1+W_2+W_3)\ge 6B \enspace .
  \end{align*}
  On the other hand, $W_1+W_2+W_3 \le W_0+W_1+W_2+W_3 = 3B$. 
  So $W_1+W_2+W_3=3B$ and thus~$W_0=0$ and all three pairwise diversities are exactly $2B$.
  Hence,
  \begin{align*}
        W_1+W_2 = 2B, \qquad
        W_1+W_3 = 2B, \qquad
        W_2+W_3 = 2B \enspace .
  \end{align*}
  Solving this set of equations gives $W_1=W_2=W_3=B$.
    
  Consequently, the jobs are partitioned into three groups $G_1,G_2,G_3$, each of total weight~$B$. 
  The dummy job has weight exactly~$B$, so the group containing the dummy job contains no other jobs. 
  The remaining two groups contain only the original~$w_j=a_j$ jobs and each group has total weight~$B$. 
  Therefore, they define a partition of $a_1, \dots, a_n$.
    
  Thus the \textsc{Partition} instance is a ``yes''-instance if and only if the constructed instance admits three feasible schedules with minimum pairwise diversity at least~$2B$. 
  Consequently, the decision version of the~$k=3$ max-min diversity problem is $\mathsf{NP}$-hard.
\end{proof}
\end{document}